\documentclass[11 pt]{article}
\usepackage{amsfonts}
\usepackage{amsmath}
\usepackage{amsthm}
\usepackage{amssymb}
\usepackage{verbatim}
\usepackage{setspace}
\usepackage{tikz}
\usepackage{tkz-berge}
\usepackage[numbers]{natbib}
\usepackage{algorithm2e}
\usepackage[top=1in, bottom=1in, left=1in, right=1in]{geometry}
\usepackage{authblk}
\usepackage{hyperref}
\usepackage{float}

\usetikzlibrary{decorations.pathmorphing, decorations.markings, snakes, positioning}

\theoremstyle{definition}

\newtheorem{theorem}{Theorem}[section]
\newtheorem{lemma}[theorem]{Lemma}

\newtheorem{proposition}[theorem]{Proposition}

\singlespace
\tikzstyle{VertexStyle}=[shape = circle,inner sep=0pt,minimum size=4pt,draw]
\tikzstyle{LabelStyle}=[font = \scriptsize\bfseries]
\tikzset{SquigglyArrow/.style={
thick,
snake=snake,
segment aspect=0,
postaction=draw}}
\SetVertexNoLabel
\begin{document}

\title{Interdiction Problems on Planar Graphs}
\author[1]{Feng Pan}
\author[2]{Aaron Schild}
\affil[1]{D-6, Los Alamos National Laboratory, \href{mailto:fpan@lanl.gov}{\texttt{fpan@lanl.gov}}}
\affil[2]{Princeton University, \href{mailto:aschild@princeton.edu}{\texttt{aschild@princeton.edu}}}

\renewcommand\Authands{ and }

\maketitle

\begin{abstract}
Interdiction problems are leader-follower games in which the leader is allowed to delete a certain number of edges from the graph in order to maximally impede the follower, who is trying to solve an optimization problem on the impeded graph. We introduce approximation algorithms and strong NP-completeness results for interdiction problems on planar graphs. We give a  multiplicative $(1 + \epsilon)$-approximation for the maximum matching interdiction problem on weighted planar graphs. The algorithm runs in pseudo-polynomial time for each fixed $\epsilon > 0$. We also show that weighted maximum matching interdiction, budget-constrained flow improvement, directed shortest path interdiction, and minimum perfect matching interdiction are strongly NP-complete on planar graphs. To our knowledge, our budget-constrained flow improvement result is the first planar NP-completeness proof that uses a one-vertex crossing gadget.
\end{abstract}

\section{Introduction}

Interdiction problems are often used to understand the robustness of solutions to combinatorial optimization problems on graphs. For any optimization problem on a graph, one can formulate an interdiction variant by creating a leader-follower game. In edge interdiction problems, every edge of the graph has an interdiction cost associated with it. The leader is given a budget and is allowed to delete any set of edges with total cost no more than the given budget. The follower then solves the given optimization problem on the remaining graph. The leader wants to pick the set of edges to delete that impedes the follower as much as possible.

In this paper, we focus on edge interdiction problems relating to shortest path interdiction, maximum flow interdiction, and matching interdiction problems. First, we give definitions for the maximum matching interdiction, budget-constrained maximum flow, shortest path interdiction, and minimum perfect matching interdiction problems.

For a given directed graph $G$ with a nonnegative edge capacities, let $\alpha(G, s, t)$ denote the value of the maximum flow from $s$ to $t$. The Budget-Constrained Flow Improvement Problem (BCFIP) \cite{Sch-98}, a problem that is closely related to the maximum flow interdiction problem, is defined as follows:\\

\textbf{Input}: A directed graph $G = (V, E)$ with a capacity function $w: E\rightarrow \mathbb{Z}_{\ge 0}$, a transport cost function $c: E\rightarrow \mathbb{Z}_{\ge 0}$, an integer budget $B > 0$, and two distinct distinguished nodes $s, t\in V$.

\textbf{Output}: A subset $I\subseteq E$ with $c(I)\le B$ that maximizes $\alpha(G[I], s, t)$, where $G[I]$ denotes the subgraph of $G$ induced by the edges in $I$.\\

Let $\alpha_B^E(G, s, t)$ denote the optimal value for BCFIP.

A \emph{matching} in a graph is a set of edges such that no two edges share an endpoint. For an edge-weighted graph $G = (V, E)$ with edge weight function $w: E\rightarrow \mathbb{Z}_{\ge 0}$, let $\nu(G)$ be the weight of a maximum weight matching in $G$. The Maximum Matching Edge Interdiction Problem (MMEIP), originally defined by Zenklusen in \cite{Zenklusen10}, is defined as follows:\\

\textbf{Input}: An edge-weighted graph $G = (V, E)$ with a weight function $w: E\rightarrow \mathbb{Z}_{\ge 0}$, an interdiction cost function $c: E\rightarrow \mathbb{Z}_{\ge 0}$, and an integer interdiction budget $B > 0$.

\textbf{Output}: A subset $I\subseteq E$ with $c(I)\le B$ that minimizes $\nu(G\backslash I)$.\\

Let $\nu_B^E(G)$ denote the optimal value for MMEIP.

Let $\rho(G, u, v)$ denote the total weight of a shortest path between $u$ and $v$ in an edge-weighted graph $G$. DSPEIP is defined as follows:\\

\textbf{Input}: An edge-weighted directed graph $G = (V, E)$ with a weight function $w: E\rightarrow \mathbb{Z}_{\ge 0}$, an interdiction cost function $c: E\rightarrow \mathbb{Z}_{\ge 0}$, an integer interdiction budget $B > 0$, and two distinct nodes $u, v\in V$.

\textbf{Output}: A subset $I\subseteq E$ with $c(I)\le B$ that maximizes $\rho(G\backslash I, u, v)$.\\

Let $\rho_B^E(G, u, v)$ denote the optimal value for DSPEIP.

A \emph{perfect matching} in a graph is a matching such that every vertex in the graph is incident with some edge in the matching. For an edge-weighted graph $G$, let $\mu(G)$ denote the weight of the minimum weight perfect matching if it exists. If no perfect matching exists in $G$, let $\mu(G) = \infty$. We introduce the Minimum Perfect Matching Edge Interdiction Problem (MPMEIP) as follows:\\

\textbf{Input}: An edge-weighted graph $G = (V, E)$ with edge weight function $w: E\rightarrow \mathbb{Z}_{\ge 0}$, interdiction cost function $c: E\rightarrow \mathbb{Z}_{\ge 0}$, and interdiction budget $B > 0$. It is assumed that, for every set $I\subseteq E$ with $c(I)\le B$, $G\backslash I$ has a perfect matching.

\textbf{Output}: A subset $I\subseteq E$ with $c(I)\le B$ that maximizes $\mu(G\backslash I)$.\\

Let $\mu_B^E(G)$ denote the optimal value for MPMEIP.

\subsection{Prior work}

Interdiction problems have practical applications in assessing the robustness of infrustructure networks. Some previous applications include drug trafficking \cite{Wood93}, military planning \cite{Ghare71}, protecting utility networks from terrorist attacks \cite{Salmeron09}, and controlling the spread of an infection \cite{Assimakopoulos87}. 

Many researchers have worked towards understanding the complexity of interdiction problems on general graphs. BCFIP is known to be strongly NP-complete on bipartite graphs and weakly NP-complete on series parallel graphs \cite{Sch-98}. Directed shortest path interdiction is strongly NP-complete on general graphs \cite{Malik89} and is strongly NP-hard to approximate within any factor better than 2 on general graphs \cite{Khachiyan08}. A node-wise variant of the shortest path interdiction problem is solvable in polynomial time \cite{Khachiyan08}. Heuristic solutions are known for shortest path interdiction. Israeli and Wood \cite{Israeli02} gave a MIP formulation and used Benders' Decomposition to solve it efficiently on graphs with up to a few thousand vertices.

The maximum flow interdiction problem is strongly NP-hard on general graphs, though a continuous variant of the problem has a pseudoapproximation that for every $\epsilon > 0$ either returns a solution within a $1 + \frac{1}{\epsilon}$-factor of the optimum or returns a better than optimal solution that uses at most $1 + \epsilon$ times the allocated budget \cite{Phillips03, Wood93}. A standard linear programming formulation of the maximum flow interdiction problem, even after adding two families of valid inequalities, has an integrality gap of $\Omega(n^{1 - \epsilon})$ for all $\epsilon\in (0, 1)$ \cite{Altner10}.

MMEIP is strongly NP-complete on bipartite graphs, even with unit edge weights and interdiction costs. Zenklusen \cite{Zenklusen10} introduced a constant factor approximation algorithm for MMEIP on graphs with unit edge weights. This algorithm makes use of iterative LP rounding. Zenklusen \cite{Zenklusen10} also showed that MMEIP is solvable in pseudo-polynomial time on graphs with bounded treewidth.

Fewer researchers have worked on interdiction problems restricted to planar graphs. Phillips \cite{Phillips93} gave a pseudo-polynomial time algorithm for the directed maximum flow edge interdiction problem. Zenklusen \cite{Zenklusen10NetworkFlow} extended this algorithm to allow for some vertex deletions and showed that the planar densest $k$-subgraph problem reduces to maximum flow interdiction with multiple sources and sinks on planar graphs. Zenklusen \cite{Zenklusen10} left the complexity of matching interdiction on planar graphs as an open problem.

\subsection{Our contributions}

\begin{figure}
\begin{center}
\begin{tabular}{|c|c|c|}
\multicolumn{3}{c}{ \bf Approximation factor lower bounds}\\ \hline
Problem name & General hardness & Planar hardness\\ \hline
Max-Flow Interdiction & Strongly NP-C & Pseudo-polynomial\cite{Phillips93}\\ \hline
DSPEIP & $(2 - \epsilon)$-inapproximable (NP-H) \cite{Khachiyan08}& \textbf{Strongly NP-C}\\ \hline
BCFIP & Strongly NP-C \cite{Sch-98} & \textbf{Strongly NP-C}\\ \hline
MPMEIP (\textbf{introduced}) & \textbf{Strongly NP-C} & \textbf{Strongly NP-C}\\ \hline
MMEIP & Strongly NP-C \cite{Zenklusen10} & \textbf{Strongly NP-C}\\ \hline
\end{tabular}
\vspace{.125 in}

\begin{tabular}{|c|c|c|}
\multicolumn{3}{c}{\bf Approximation factor upper bounds}\\ \hline
Problem name & General approximation & Planar approximation\\ \hline
Max-Flow Interdiction & Open & 1\cite{Phillips93}\\ \hline
DSPEIP & Open & Open\\ \hline
BCFIP & Open & Open\\ \hline
MPMEIP (\textbf{introduced}) & Open & Open\\ \hline
MMEIP & O(1) weighted \cite{Dinitz13}, 4 unweighted\cite{Zenklusen10} & \textbf{Pseudo-PTAS}\\ \hline
\end{tabular}
\end{center}

\caption{A summary of  known results about the problems we consider. Our results are displayed in bold. ``Open'' means that no nontrivial approximation algorithm is known. ``NP-C'' abbreviates ``NP-complete.'' ``NP-H'' abbreviates ``NP-hard.''}
\end{figure}

In this paper, we give a pseudo-polynomial time approximation scheme for MMEIP on undirected planar graphs. A \textit{pseudo-polynomial time approximation scheme} (Pseudo-PTAS) is an algorithm that takes a parameter $\epsilon\in (0, 1)$ as additional input and outputs a solution with objective value within a multiplicative $(1 + \epsilon)$-factor of the optimum (or a $(1 - \epsilon)$ factor for maximization problems). Furthermore, the algorithm terminates in polynomial time for fixed $\epsilon$. In Section \ref{PTASMaxInt}, we give a pseudo-polynomial  algorithm that achieves the following guarantee:

\begin{theorem}\label{PTASMaxMatchingThm}
Let $I$ be an edge set with $c(I)\le B$ that minimizes $\nu(G\backslash I)$. There is an algorithm that, for every $\epsilon > 0$, returns a set $\widehat{I}$ of edges for which $$\nu(G\backslash \widehat{I})\le (1 + \epsilon)\nu(G\backslash I)$$ and $c(\hat{I})\le B$. Furthermore, for every fixed $\epsilon$, the algorithm runs in polynomial time with respect to the size of the graph, the sum of the edge weights, and the sum of the interdiction costs of all edges (pseudo-polynomial time).
\end{theorem}

To obtain our Pseudo-PTAS for MMEIP, we extend Baker's technique \cite{Baker94} for interdiction problems on planar graphs. To our knowledge, our algorithm is the first using Baker's technique in solving any network interdiction problem. With Baker's technique, MMEIP can be solved approximately in pseudo-polynomial time by using using a tree decomposition with bounded width. We give the approximation ratio and run-time analysis in Section \ref{PTASMaxInt}.

A pseudo-polynomial algorithm is known for solving maximum flow interdiction on directed planar graphs. \cite{Phillips93} No pseudo-polynomial algorithms for solving other network interdiction problems on planar graphs are known. We explain this nonexistence by showing that BCFIP, DSPEIP, MPMEIP, and MMEIP are strongly NP-complete even on planar graphs. In Sections \ref{MultiMaxFlow}, \ref{ShortPathInt}, \ref{MinPerfectInt}, and \ref{MaxInt} respectively, we prove the following results:

\begin{theorem}\label{MaxFlowNPC}
It is strongly NP-complete to decide, given an integer budget $B > 0$ and an integer $k \ge 0$, whether or not $\alpha_B^E(G, s, t) > k$, even when $G$ is a directed planar graph and $s$ and $t$ border a common face.
\end{theorem}

\begin{theorem}\label{ShortestPathNPC}
Given an edge-weighted directed planar graph $G$, it is strongly NP-complete to decide, given an integer budget $B > 0$ and an integer $k \ge 0$, whether or not $\rho_B^E(G, u, v) > k$.
\end{theorem}

\begin{theorem}\label{MinPerfectMatchingNPC}
Given an edge-weighted undirected bipartite planar graph $G$, it is strongly NP-complete to decide, given an integer budget $B > 0$ and an integer $k \ge 0$, whether or not $\mu_B^E(G) > k$.
\end{theorem}

\begin{theorem}\label{MaxMatchingNPC}
Given an edge-weighted undirected bipartite planar graph $G$, it is strongly NP-complete to decide, given an integer budget $B > 0$ and an integer $k \ge 0$, whether or not $\nu_B^E(G) < k$.
\end{theorem}

The strong NP-completeness of MMEIP on planar graphs implies that our Pseudo-PTAS is optimal with respect to approximation ratio, resolving the question asked by Zenklusen \cite{Zenklusen10} about the complexity of MMEIP on planar graphs.

We introduce a novel crossing removal technique to show that the BCFIP problem is strongly NP-complete on directed planar graphs. Prior NP-completeness results on planar graphs (e.g. \cite{Garey76, Garey77, Stockmeyer73}) reduce problems on general graphs to their counterparts on planar graphs by replacing edge crossings with a constant-size graph called a \textit{crossing gadget}. While crossing gadgets are known for many different problems, they change the numerical value of the solution and, therefore, cannot be approximation-preserving. To our knowledge, we introduce the first simple crossing gadget that preserves the optimal value.

We then show that the reduction from maximum flow interdiction (with source and sink on the same face) to multi-objective shortest path on planar graphs \cite{Phillips93} also gives an approximation-preserving reduction from $s-t$ planar BCFIP to shortest path interdiction on directed planar graphs. A well-known reduction from the shortest path problem to the assignment problem \cite{Hoffman63} fails to preserve planarity. We give a new approximation-preserving reduction that reduces directed shortest path interdiction to minimum perfect matching interdiction on planar graphs. Finally, we reduce minimum perfect matching interdiction to MMEIP on planar graphs using edge weight manipulations. Our results distinguish BCFIP, shortest path interdiction, minimum perfect matching interdiction, and MMEIP from the maximum flow interdiction problem, which is solvable in pseudo-polynomial time on planar graphs.

In Section \ref{Preliminaries}, we give some notation that we use throughout this paper. In Section \ref{PTASMaxInt}, we present our Pseudo-PTAS for MMEIP. In Section \ref{MultiMaxFlow}, we prove that BCFIP is strongly NP-complete on planar graphs. In Section \ref{ShortPathInt}, we reduce BCFIP to directed shortest path interdiction on planar graphs. In Section \ref{MinPerfectInt}, we reduce directed shortest path interdiction to minimum perfect matching interdiction. In Section \ref{MaxInt}, we reduce minimum perfect matching interdiction to MMEIP. We conclude with Section \ref{Conclusion}, in which we discuss open problems relating to interdiction.

\section{Preliminaries}\label{Preliminaries}

For a (undirected or directed) graph $G$, let $V(G)$ denote its vertex set and $E(G)$ denote its edge set. Edges in a directed graph are denoted by ordered pairs $(u, v)$ for $u, v\in V(G)$, while edges in an undirected graph are denoted by unordered pairs $\{u, v\}$. For a (undirected or directed) graph $G$, let $G[U]$ denote the subgraph induced by the vertices in $U\subseteq V(G)$, i.e. the subgraph of $G$ for which $V(G[U]) = U$ and $E(G[U]) = \{\{u, v\}\in E(G): u, v\in U\}$ (undirected graph, for a directed graph replace $\{u, v\}$ with $(u, v)$). For $F\subseteq E(G)$, let $G[F]$ denote the subgraph of $G$ induced by the edges in $F$, i.e. the graph with $V(G[F]) = \{v\in V(G): \exists w | \{w, v\}\in F\}$ (undirected) or $V(G[F]) = \{v\in V(G): \exists  (w, v)\in F\vee (v, w)\in F\}$ (directed) and $E(G[F]) = F$.

For a set $S\subseteq V(G)$, let $\delta(S)\subseteq E(G)$ be the set of edges with exactly one endpoint in $S$. For a directed graph, let $\delta^+(S) = \{(u, v)\in E(G): u\in S, v\notin S\}$ and let $\delta^-(S) = \{(u, v)\in E(G): u\notin S, v\in S\}$. Note that $\delta(S) = \delta^-(S)\cup \delta^+(S)$ for any $S\subseteq V(G)$ if $G$ is directed. For a vertex $v\in V(G)$, let $\delta(v) = \delta(\{v\})$, $\delta^+(v) = \delta^+(\{v\})$, and $\delta^-(v) = \delta^-(\{v\})$. For a directed graph $G$, let $s, t: E(G)\rightarrow V(G)$ denote $s((u, v)) = u$ and $t((u, v)) = v$ for all $(u, v)\in E(G)$. For any function $f: E(G)\rightarrow \mathbb{R}$ and $F\subseteq E(G)$, let $f(F) = \sum_{e\in F} f(e)$.

\subsection{Tree decompositions and treewidth}
For a set $S$, let $\mathcal{P}(S)$ denote its power set. For an undirected graph $G$, a tree decomposition \cite{RS84, Halin76} of $G$ is a pair $(T, \mathcal{B}: V(T)\rightarrow \mathcal{P}(V(G)))$ where $T$ is a tree. Furthermore, $\mathcal{B}$ has the following properties:

\begin{itemize}
\item $\cup_{w\in V(T)}\mathcal{B}(w) = V(G)$
\item For all $\{u, v\}\in E(G)$, there is some $w\in V(T)$ such that $u, v\in \mathcal{B}(w)$.
\item For any $v\in V(G)$, let $U_v\subseteq V(T)$ be the set of vertices $w\in V(T)$ for which $v\in \mathcal{B}(w)$. Then, $T[U_v]$ is connected.
\end{itemize}

Let $k_T = \max_{w\in V(T)} |\mathcal{B}(w)| - 1$ denote the \textit{width} of $T$. The treewidth of $G$ is the minimum width of any tree decomposition of $G$.

\subsection{Planar graphs}
In this paper, we use the combinatorial definition of planar graphs given in \cite{GraphsOnSurfaces}. First, we define undirected planar graph embeddings. Consider a set of objects $E$, which can be thought of as a set of undirected edges. The set of \emph{darts} is defined as $F = E\times \{-1, 1\}$. Define a graph $G$ from $E$ as a pair $(\pi, F)$ where $\pi$ is a permutation on the edge set. $\pi$ is a product of cycles. Let $V(G)$, the vertex set of $G$, be the set of these cycles.

Define three maps $\text{head}_G: F\rightarrow V(G)$, $\text{tail}_G: F\rightarrow V(G)$, and $\text{rev}: F\rightarrow F$ as follows. Let $\text{rev}((e, x)) = (e, -x)$. Note that each dart $d\in F$ is in a unique cycle of $V(G)$. Let $\text{head}_G(d)$ be the unique cycle containing $d$ and let $\text{tail}_G(d) = \text{tail}_G(\text{rev}(d))$. An undirected planar graph is also an undirected graph with vertex set $V(G)$ and edge set $E(G) = \{\{\text{tail}_G(d), \text{head}_G(d)\}: \forall d\in F\}$ and let $E_G: F\rightarrow E(G)$ be defined by $E_G(d) = \{\text{tail}_G(d), \text{head}_G(d)\}$ for any $d\in F$.

The \emph{dual graph} of $G$ is the graph $G^*:= (\pi\circ\text{rev}, F)$. For each edge $e$, there are precisely two darts $d_1$ and $d_2$ such that $E_G(d_1) = E_G(d_2) = e$. Let $D_G: E(G)\rightarrow E(G^*)$ be the bijection defined by $D_G(e) := E_{G^*}(d_1)$. Note that $E_{G^*}(d_1) = E_{G^*}(d_2)$ since $d_1 = \text{rev}(d_2) = \text{rev}(d_2)$. $G$ is a \emph{undirected planar graph} if $G$ satisfies Euler's formula, i.e. $|V(G)| - |E(G)| + |V(G^*)| = 2$. The set of \emph{faces} of $G$ is defined to be $F(G) := V(G^*)$. Define $B_G: F(G)\rightarrow 2^{V(G)}$ by $B_G(f) := V(G[D_G^{-1}(\delta(f))])$ and call $B_G(f)$ the set of \emph{boundary vertices} of $f$.

Now, we define directed planar graphs. Let $E$ and $E'$ be sets of objects with $E\subseteq E'$ and let $R: E'\rightarrow E'$ be a bijection with $R\circ R = \text{id}$ and $R(e')\ne e'$ for all $e'\in E'$. $E$ can be thought of as a set of directed edges of a graph, $E'$ can be thought of as the edge set of the undirected version of the graph, and $R$ can be thought of as the map from an edge to its reverse direction copy. $R$ induces a perfect matching on the elements of $E'$. A \emph{directed planar graph} $G$ is a tuple $(\pi, E, E', R)$, where $\pi$ is a permutation on the elements of $E'$. Furthermore, there is a bijection $b$ between $E'$ and a set of darts $F$ such that $b(R(e')) = \text{rev}(b(e'))$ for all $e'\in E'$, where $H = (\pi, b(E'))$. We put a further constraint on $G$ that $(\pi, b(E'))$ must be an undirected planar graph. Let $V(G) := V((\pi, b(E')))$ and $E(G) := \{(\text{tail}_G(b(e), \text{head}_G(b(e))): \forall e\in E\}$.

The \emph{dual graph} $G*$ of a directed planar graph $G$ is the directed planar graph $(\pi\circ\text{rev}, E, E', R)$. There is a bijection $c_G: E\rightarrow E(G)$ defined by $c_G(e) = (\text{tail}_G(b(e), \text{head}_G(b(e)))$ for all $e\in E$. Let $D_G: E(G)\rightarrow E(G^*)$ be the map $D_G = c_{G^*}\circ c_G^{-1}$. Define the set of faces $F(G) := F(H)$ and the boundary vertices by $B_G := B_H$.

Since $\text{rev}\circ \text{rev} = \text{id}$, the identity permutation, $(G^*)^* = G$. One can also show the following propositions:

\begin{proposition}\label{cut-cycle-duality}
In an undirected planar graph $G$, let $C_{s, t}\subseteq V(G)$ be a set such that $\delta(C_{s, t})$ is a minimum $s-t$ cut with respect to some edge capacity function $c:E(G)\rightarrow \mathbb{R}_{\ge 0}$. Furthermore, suppose that $s, t\in B_G(f_{\infty})$ for some $f_{\infty}\in F(G)$. Then, $D_G(\delta(C_{s, t}))$ is a simple cycle in $G^*$ that passes through $f_{\infty}$.
\end{proposition}

\begin{proposition}\label{boundary-deletions}
Consider an undirected planar graph $G = (\pi, F)$ and consider a face $f$. $f$ is a cycle $C = (d_1, d_2, \hdots, d_r)$ of the permutation $\pi\circ\text{rev}$. Let $t_i := \text{tail}_G(d_i) = \text{head}_G(d_{i-1})$. Suppose that there are indicies $i$ and $j$, $i\le j$ with $i-j\ne 1, 0, -1\pmod r$ such that $t_i = t_j$. Then $G\backslash t_i$ is disconnected.
\end{proposition}

\begin{proof}
Since $i-j\ne 1, 0, -1\pmod r$, there are vertices $t_a$ and $t_b$ along the boundary cycle with $a\in [i, j)$ and $b\notin [i, j)$ with $a\neq i$, $b\neq j$. Create a new graph by replacing $t_i$ with two vertices $u$ and $v$ connected by an edge $e$. Call this graph $G' = (\pi', F')$. Note that this graph still satisfies Euler's formula, as the number of vertices increased by one, the number of edges increased by one, and the number of faces did not change. Deleting $e$ from $G'$ does not change the number of faces or vertices. However, it does decrease the number of edges, so $G'\backslash e$ violates Euler's formula. Therefore, it must be disconnected and any path between $t_a$ and $t_b$ in $G'$ must contain $e$. Contracting $e$ to obtain $G$ shows that any path between $t_a$ and $t_b$ contains $t_i = t_j$. Therefore, $t_a$ cannot be connected to $t_b$ in $G\backslash t_i$.
\end{proof}

For an integer $k > 0$, call an undirected planar graph $G$ a \emph{$k$-outerplanar graph} if there is a face $f^{\infty}\in F(G)$ such that for every vertex $v\in V(G)$, there is a path containing at most $k-1$ edges to a vertex $w\in B_G(f^{\infty})$. Bodlaender \cite{Bodlaender88} showed the following theorem:

\begin{theorem}\label{KOuterplanarTW}
$k$-outerplanar graphs have treewidth at most $3k - 1$. Furthermore, a tree decomposition of width at most $3k - 1$ can be found in $O(kn)$ time.
\end{theorem}

\subsection{Line graphs}
Let $G$ be an undirected graph. The \emph{line graph} $L_u(G)$ is an undirected graph with $V(L_u(G)) := E(G)$ and $E(L_u(G)) := \{\{e_1, e_2\}: \forall e_1, e_2\in E(G) \text{ $e_1$ and $e_2$ share an endpoint}\}$. For a directed graph $H$, the \emph{line graph} $L_d(H)$ is defined by $V(L_d(H)) = E(H)$ and $E(L_d(H)) = \{(e_1, e_2): \forall e_1=(u_1, v_1), e_2=(u_2, v_2)\in E(H), v_1 = u_2\}$. One can show the following proposition:

\begin{proposition}\label{planar-directed-line-graph}
The directed line graph of a directed planar graph with maximum degree 3 is a directed planar graph.
\end{proposition}

\section{A Pseudo-PTAS for maximum matching interdiction on planar graphs}\label{PTASMaxInt}

In this section, we introduce a pseudo-polynomial time approximation scheme for the maximum matching interdiction problem (MMEIP). We use many ideas from Baker's framework introduced in \cite{Baker94}. First, we give our Pseudo-PTAS. We will use Zenklusen's algorithm for bounded treewidth graphs in \cite{Zenklusen10} as a subroutine.
\vspace{.125 in}

\begin{figure}
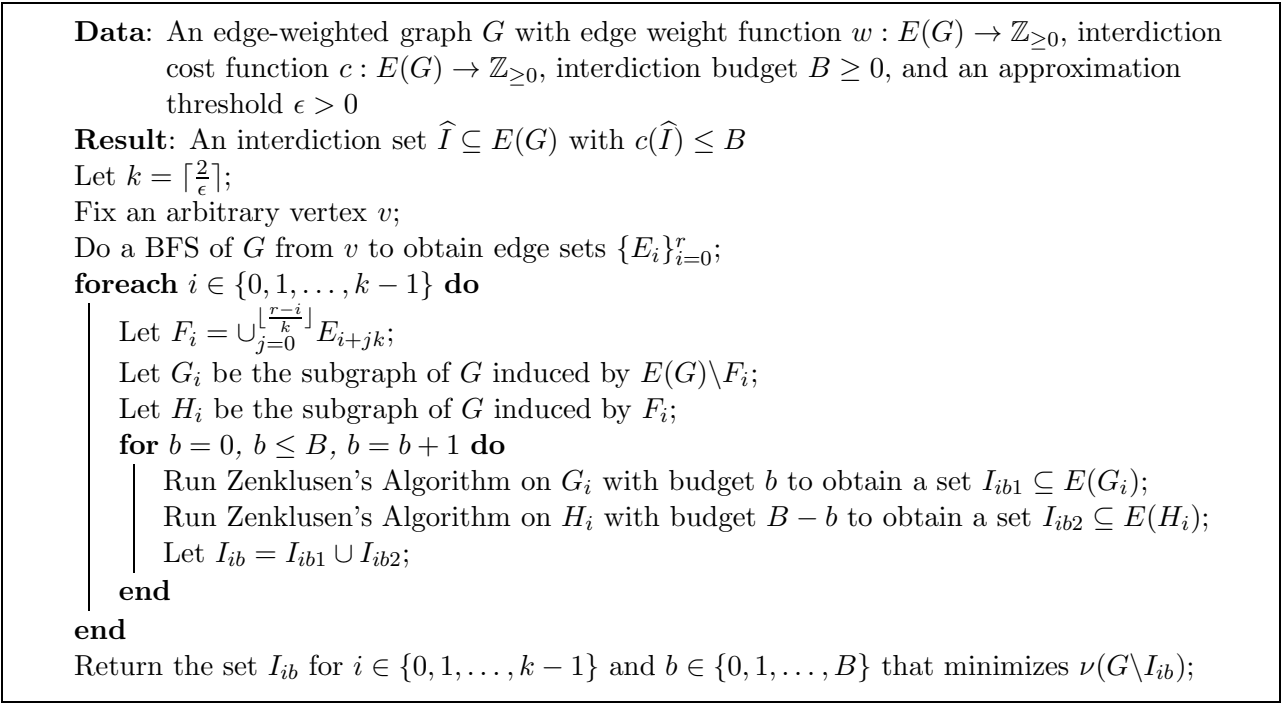

\fbox{
\begin{algorithm}[H]
\SetAlgoLined
 \KwData{An edge-weighted graph $G$ with edge weight function $w: E(G)\rightarrow \mathbb{Z}_{\ge 0}$, interdiction cost function $c: E(G)\rightarrow \mathbb{Z}_{\ge 0}$, interdiction budget $B\ge 0$, and an approximation threshold $\epsilon > 0$}
 \KwResult{An interdiction set $\widehat{I}\subseteq E(G)$ with $c(\widehat{I})\le B$}
 Let $k = \lceil\frac{2}{\epsilon}\rceil$\;
 Fix an arbitrary vertex $v$\;
 Do a BFS of $G$ from $v$ to obtain edge sets $\{E_i\}_{i = 0}^r$\;
 \ForEach{$i\in \{0, 1, \hdots, k - 1\}$}{
   Let $F_i = \cup_{j = 0}^{\lfloor\frac{r - i}{k}\rfloor} E_{i + jk}$\;
   Let $G_i$ be the subgraph of $G$ induced by $E(G)\backslash F_i$\;
   Let $H_i$ be the subgraph of $G$ induced by $F_i$\;
   \For{$b = 0$, $b\le B$, $b = b + 1$}{
    Run Zenklusen's Algorithm on $G_i$ with budget $b$ to obtain a set $I_{ib1}\subseteq E(G_i)$\;
    Run Zenklusen's Algorithm on $H_i$ with budget $B - b$ to obtain a set $I_{ib2}\subseteq E(H_i)$\;
    Let $I_{ib} = I_{ib1}\cup I_{ib2}$\;
   }
 }
 Return the set $I_{ib}$ for $i\in \{0, 1, \hdots, k - 1\}$ and $b\in \{0, 1, \hdots, B\}$ that minimizes $\nu(G\backslash I_{ib})$\;
\label{MatchingAlgorithm}
\end{algorithm}
}
\caption{Algorithm for MMEIP on planar graphs}
\end{figure}
\vspace{.125 in}

Now, we prove that the algorithm produces a good interdiction set. First, we show three properties of maximum weight matchings, which encapsulate our use of the structure of matchings to obtain a Pseudo-PTAS:

\begin{proposition}\label{subgraph}
Let $G$ be an edge-weighted graph with weight function $w: E(G)\rightarrow \mathbb{Z}_{\ge 0}$ and let $H$ be a subgraph of $G$. Then, $\nu(H)\le \nu(G)$.
\end{proposition}

\begin{proof}
All matchings in $H$ are matchings in $G$, so the result follows.
\end{proof}
\vspace{.125 in}

\begin{proposition}\label{triangle}
Let $G$ be an edge-weighted graph with weight function $w: E(G)\rightarrow \mathbb{Z}_{\ge 0}$. Consider any $F\subseteq E(G)$. Let $G_1$ be the subgraph of $G$ induced by the edge set $F$ and let $G_2$ be the subgraph of $G$ induced by the edge set $E(G)\backslash F$. Then, $\nu(G)\le \nu(G_1) + \nu(G_2)$.
\end{proposition}

\begin{proof}
Consider a maximum weight matching $M$ of $G$. Note that $M\cap E(G_1)$ is a matching in $G_1$ and that $M\cap E(G_2)$ is a matching of $G_2$. Since $E(G_1)\cap E(G_2) = \emptyset$, $$\nu(G) = w(M) = w(M\cap E(G_1)) + w(M\cap E(G_2))\le \nu(G_1) + \nu(G_2)$$ as desired.
\end{proof}
\vspace{.125 in}

\begin{proposition}\label{edge-sum}
Let $G$ be an edge-weighted graph with weight function $w: E(G)\rightarrow \mathbb{Z}_{\ge 0}$. Pick any node $v\in V(G)$ and do a breadth-first search starting from $v$ and label the vertices of $G$ with their distance from $v$. Let $E_i$ be the set of edges from vertices with label $i$ to vertices with label $i + 1$. Fix an integer $k > 1$. For all $i\in \{0, 1, \hdots, k - 1\}$, let $F_i = \cup_{j = 0}^{\lfloor\frac{r - i}{k}\rfloor} E_{i + jk}$, where $r$ is the number of labels, and let $G_i$ be the subgraph of $G$ induced by $F_i$. Then, $$\sum_{i = 0}^{k - 1} \nu(G_i)\le 2\nu(G)$$
\end{proposition}

\begin{proof}
Note that for $k > 1$, no edge in $E_{i + jk}$ shares an endpoint with any edge in $E_{i + j'k}$ for $j\ne j'$. Therefore, if $\widehat{G}_i$ denotes the subgraph of $G$ induced by $E_i$, then $$\nu(G_i) = \sum_{j = 0}^{\lfloor\frac{r - i}{k}\rfloor} \nu(\widehat{G}_{i + jk})$$ Let $F' = \cup_{j = 0}^{\lfloor\frac{r}{2}\rfloor} E_{2j}$ and let $F'' = \cup_{j = 0}^{\lfloor\frac{r - 1}{2}\rfloor} E_{2j + 1}$. Let $G'$ be the subgraph of $G$ induced by $F'$ and let $G''$ be the subgraph of $G$ induced by $F''$. Then, reindexing shows that $$\sum_{i = 0}^{k - 1} \nu(G_i) = \sum_{j = 0}^r \nu(\widehat{G}_j) = \nu(G') + \nu(G'')\le 2\nu(G)$$
\end{proof}
\vspace{.25 in}

We prove Theorem \ref{PTASMaxMatchingThm} by proving two theorems about the performance of Algorithm \ref{MatchingAlgorithm}. We restate them here:

\begin{theorem}[Approximation guarantee for Algorithm \ref{MatchingAlgorithm}]
Let $I$ be an edge set with $c(I)\le B$ that minimizes $\nu(G\backslash I)$. The set $\widehat{I}$ returned by Algorithm \ref{MatchingAlgorithm} satisfies $$\nu(G\backslash \widehat{I})\le (1 + \epsilon)\nu(G\backslash I)$$
\end{theorem}

\begin{proof}
By Proposition \ref{edge-sum} applied to $G\backslash I$, there is some $i$ such that $\nu(H_i\backslash I)\le \frac{2}{k}\nu(G\backslash I)$. Let $b = c(I\backslash F_i)$. Since $c(I)\le B$, $c(I\cap F_i)\le B - b$. Zenklusen's algorithm returns the optimal interdiction sets on $G_i$ and $H_i$ for the budgets $b$ and $B - b$ respectively. Therefore, $\nu(G_i\backslash I_{ib1})\le \nu(G_i\backslash I)$ and $\nu(H_i\backslash I_{ib2})\le \nu(H_i\backslash I)$. By Proposition \ref{subgraph}, $\nu(G_i\backslash I)\le \nu(G\backslash I)$. Summing inequalities and applying Proposition \ref{triangle} shows that

\begin{align*}
\nu(G\backslash I_{ib})&\le \nu(G_i\backslash I_{ib1}) + \nu(H_i\backslash I_{ib2})\\
&\le \nu(G_i\backslash I) + \nu(H_i\backslash I)\\
&\le (1 + \frac{2}{k})\nu(G\backslash I)\\
&\le (1 + \epsilon)\nu(G\backslash I)
\end{align*}

Since $\nu(G\backslash \widehat{I})\le \nu(G\backslash I_{ib})$, we are done.
\end{proof}
\vspace{.25 in}

\begin{theorem}[Runtime guarantee for Algorithm \ref{MatchingAlgorithm}]
For fixed $\epsilon$, this algorithm terminates in pseudo-polynomial time on planar graphs and, more generally, apex-minor-free graphs. More precisely, on planar graphs, it terminates in time $$O((2B/\epsilon)(|V(G)|(C + 1)^{8^{2/\epsilon + 1}}) + (2B/\epsilon)|E(G)|\sqrt{|V(G)|})$$ where $C = w(E(G))$.
\end{theorem}

\begin{proof}
Zenklusen's algorithm \cite{Zenklusen10} has runtime $O(|V(G)|B^2(C + 1)^{2^{t + 1}})$ on graphs with treewidth at most $t$. The breadth-first search at the beginning of the algorithm takes $O(|E| + |V|)$ time. The innermost for loop is run $B + 1$ times per execution of the outermost for loop, which runs at most $k \le \frac{2}{\epsilon} + 1$ times. By Theorem \ref{KOuterplanarTW}, the treewidths of $G_i$ and $H_i$ are at most $3k - 1$ and $2$ respectively. The computation on $G_i$ dominates the computation on $H_i$. Using the Hopcroft-Karp algorithm \cite{HopcroftK71} to find the best interdiction set on the last line of Algorithm \ref{MatchingAlgorithm} gives the last term of the runtime.
\end{proof}

\section{Strong NP-hardness of budget-constrained flow improvement on directed planar graphs}\label{MultiMaxFlow}

In this section, we prove Theorem \ref{MaxFlowNPC}. We reduce from the maximum independent set problem on general graphs. Before proving Theorem \ref{MaxFlowNPC}, we discuss the proof of the following easier result, which was shown using a different reduction in \cite{Sch-98}:

\begin{lemma}\label{WeakMaxFlowNPC}
The decision version of BCFIP is strongly NP-complete on general directed graphs.
\end{lemma}

\begin{proof}
Consider an undirected graph $G$ for which we want to decide if there is an independent set in $G$ of size at least $k$. Create a graph $G_1$ with one vertex for every vertex of $G$, one vertex for every edge of $G$, and two distinguished vertices $s$ and $t$. The edges in $G_1$ will be split into four sets $E_1$, $E_2$, $E_3$, and $E_4$. There will be directed edges in $G_1$ of the following types:\\

\begin{enumerate}
\item One edge from $s$ to each vertex $v^v$ corresponding to a vertex of $G$. ($E_1$)
\item For every edge $e = \{u, v\}\in E(G)$, two directed edges $(v^u, e)$ and $(v^v, e)$ in $E(G_1)$. ($E_2$) 
\item One edge from each vertex $v^e$ corresponding to an edge of $G$ to $t$. ($E_3$)
\item One edge from each vertex $v^v$ corresponding to a vertex of $G$ to $t$. ($E_4$)
\end{enumerate}

Let $d$ be the maximum degree of any node in $G$. The capacity function $w_1: E(G_1)\rightarrow \mathbb{Z}_{\ge 0}$ will be defined as follows:\\

\begin{displaymath}
   w_1(e) = \left\{
     \begin{array}{lr}
       d & : e \in E_1\\
       1 & : e \in E_2\\
       1 & : e \in E_3\\
       d - \text{deg}_G(\eta(e)) & : e \in E_4
     \end{array}
   \right.
\end{displaymath}

where $\eta: E_4\rightarrow V(G)$ returns the vertex of $G$ corresponding to the left endpoint of the input edge of $G_1$. The cost function $c_1: E(G_1)\rightarrow \mathbb{Z}_{\ge 0}$ will be defined as follows:

\begin{displaymath}
   c_1(e) = \left\{
     \begin{array}{lr}
       1 & : e \in E_1\\
       0 & : e \notin E_1\\
     \end{array}
   \right.
\end{displaymath}

This construction is depicted in Figure \ref{BCFIPNonplanarImage}. The proof of the following proposition implies the lemma:

\begin{proposition}\label{NonplanarMaxFlowNPC}
$\alpha_k^E(G_1, s, t) = kd$ if and only if there is an independent set in $G$ of size at least $k$.
\end{proposition}

\begin{proof}[Proof of the if direction]
Suppose that $G$ has an independent set $S\subseteq V(G)$ with $|S| \ge k$ and let $S'$ be the corresponding set of vertices in $G_1$. The following function $f: E(G)\rightarrow \mathbb{Z}_{\ge 0}$ is a valid flow on $G_1$:

\begin{displaymath}
   f(\{v_1, v_2\}) = \left\{
     \begin{array}{ll}
       d & : \{s, v_2\} \in E_1, v_2\in S'\\
       1 & : \{v_1, v_2\} \in E_2, v_1\in S'\\
       1 & : \{v_1, t\} \in E_3, v_1\in \delta(S) \text{ ($v_1$ corresponds to an edge in $G$ here)}\\
       d - \text{deg}_G(\eta(e)) & : \{v_1, t\} \in E_4, v_1\in S'\\
       0 & : \text{otherwise}
     \end{array}
   \right.
\end{displaymath}

The cost of the nonzero edges in this flow is $k$, as there are precisely $k$ edges in $E_1$ with positive flow value. The value of this flow is $kd$, so $\alpha_k^E(G_1, s, t) \ge kd$. Since at most $k$ edges from $E_1$ can be paid for, there is a cut of total capacity $kd$ in the subgraph of $G$ induced by paid-for edges. Therefore, $\alpha_k^E(G_1, s, t)\le kd$ and this direction is complete.
\end{proof}

\begin{proof}[Proof of the only if direction]
Suppose that $\alpha_k^E(G_1, s, t) = kd$. Let $f: E(G_1)\rightarrow \mathbb{Z}_{\ge 0}$ be a flow with flow value $kd$. Since at most $k$ edges of $E_1$ can be paid for, $f(e) = 0$ or $f(e) = d$ for all $e\in E_1$. Furthermore, exactly $k$ edges$e$ in $E_1$ have the property that $f(e) = d$. The right endpoints of these $k$ edges correspond to vertices in $G$ and are denoted as set $T$. Next, we show that $T$ is an independent set in $G$.\\

Suppose instead that $T$ is not an independent set and that for two vertices $u, v\in T$, $e = \{u, v\}\in E(G)$. Let $u'$ and $v'$ be the corresponding vertices to $u$ and $v$ respectively in $V(G_1)$ and let $e'$ be the vertex in $V(G_1)$ corresponding to $e$. Since $f(\{s, u'\}) = f(\{s, v'\}) = d$, which is full flow, all outgoing edges from $u'$ and $v'$ must have full flow. This means that $f(\{u', e'\}) = f(\{v', e'\}) = 1$ and that $f(\{e', t\}) = 2 > w_1(\{e', t\}) = 1$, a contradiction. Therefore, $T$ is an independent set in $G$ of size at least $k$.
\end{proof}

Since this reduction takes polynomial time, Lemma \ref{WeakMaxFlowNPC} is proven.
\end{proof}

\begin{figure}[H]

\begin{tikzpicture}
\node(MaxFlowA){
	\begin{tikzpicture}[scale = 2]
		\Vertex[x = 0, y = 0]{A0}
		\Vertex[x = 0, y = 2]{A1}
		\Vertex[x = 1, y = 1]{A2}
		\Vertex[x = 2, y = 1]{A3}

		\Edge[label = a](A0)(A1)
		\Edge[label = b](A0)(A2)
		\Edge[label = c](A1)(A2)
		\Edge[label = d](A2)(A3)
	\end{tikzpicture}
};
\node[yshift = -10cm](MaxFlowB){
	\begin{tikzpicture}[scale = 1.1]
		\tikzset{VertexStyle/.append style = {fill = blue}}
		\Vertex[x = 0, y = 5]{S}
		\Vertex[x = 12, y = 5]{T}
		\tikzset{VertexStyle/.append style = {fill = none}}

		\Vertex[x = 4, y = 8]{A0}
		\Vertex[x = 4, y = 6]{A1}
		\Vertex[x = 4, y = 4]{A2}
		\Vertex[x = 4, y = 2]{A3}

		\Vertex[x = 8, y = 8]{a}
		\Vertex[x = 8, y = 6]{b}
		\Vertex[x = 8, y = 4]{c}
		\Vertex[x = 8, y = 2]{d}

		\tikzset{EdgeStyle/.style={ 
			postaction=decorate,
            		decoration={
                			markings,
                			mark=at position 1 with {\arrow[scale = 1.5]{>}}
            		}
		}}
		\Edge[style = dashed, label = $3|1$](S)(A0)
		\Edge[style = dashed, label = $3|1$](S)(A1)
		\Edge[style = dashed, label = $3|1$](S)(A2)
		\Edge[style = dashed, label = $3|1$](S)(A3)

		\Edge[label = {$1|0$,a}](a)(T)
		\Edge[label = {$1|0$,b}](b)(T)
		\Edge[label = {$1|0$,c}](c)(T)
		\Edge[label = {$1|0$,d}](d)(T)

		\Edge(A0)(a)
		\Edge(A1)(a)

		\Edge(A0)(b)
		\Edge(A2)(b)

		\Edge(A1)(c)
		\Edge(A2)(c)

		\Edge(A2)(d)
		\Edge(A3)(d)

		\Edge[color = red, label = $1|0$](A0)(T)
		\Edge[color = red, label = $1|0$](A1)(T)
		\Edge[color = red, label = $2|0$](A3)(T)
	\end{tikzpicture}
};

\path[SquigglyArrow, ->](MaxFlowA)--(MaxFlowB);
\end{tikzpicture}


\caption{The reduction in the proof of Lemma \ref{WeakMaxFlowNPC}, with $w|c$ as the edge labels. Dashed edges must be paid for, while dark edges are never deleted (since they have cost 0). All edges in $E_2$ (which are the unlabeled edges here) have label $1|0$.}\label{BCFIPNonplanarImage}
\end{figure}
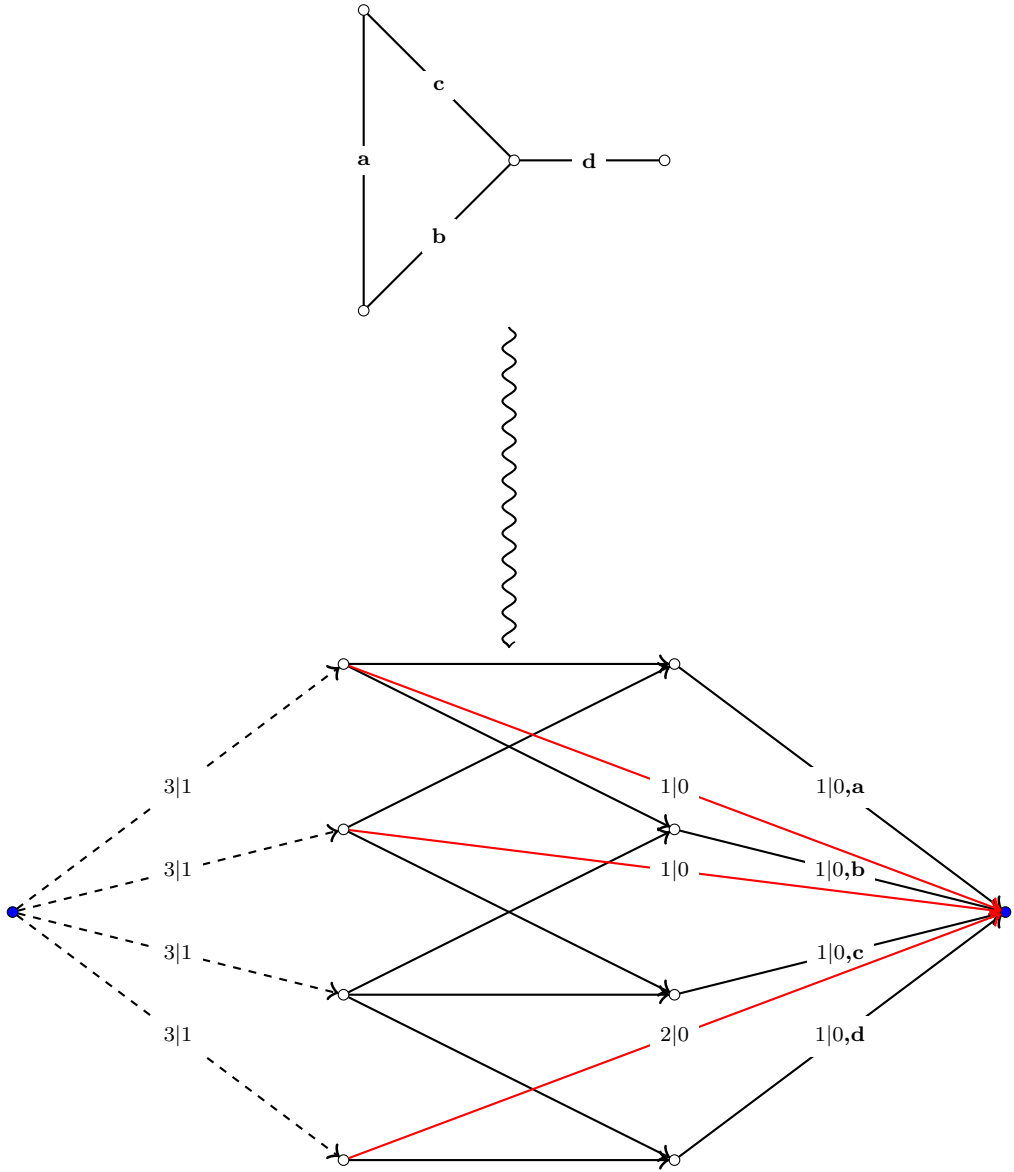

There are two properties of this reduction that facilitate crossing removal:

\begin{enumerate}
\item If $G$ has an independent set of size $k$, then there exists a set $F\subseteq E(G_1)$ that will give rise to an optimal flow with full flow along all edges in $F$.
\item If $G$ has an independent set $T$ of size $k$, $\delta^+(v)\cup \delta^-(v)\subseteq F$ for every vertex $v\in V(G_1)$ corresponding to a vertex of $T$.
\end{enumerate}

We will prove Theorem \ref{MaxFlowNPC} by removing crossings from $G_1$.  First, we will  discuss the key idea. The first property states that the existence of an independent set of size at least $k$ guaranteed that all edges transfered either no flow or full flow in. If two edges cross and both have full flow,  adding a vertex at their intersection does not alter flow from the original network. If two edges cross and only one edge has full flow, we need to ensure the flow is not redistributed by adding a vertex at the crossing. We will show that there exists a combination of edge capacities and costs to guarantee that  
no redistribution of flow can result in maximum flow.

We will now describe how to remove crossings from $G_1$ to obtain a planar graph $G_2$ with an associated capacity function $w_2: E(G_2)\rightarrow \mathbb{Z}_{\ge 0}$ and cost function $c_2: E(G_2)\rightarrow \mathbb{Z}_{\ge 0}$:\\

\begin{enumerate}
	\item Embed $G_1$ in an $x-y$ coordinate system so that it has the following properties:\label{G1Embedding}
	\begin{itemize}
		\item $s$ has coordinates $(0, 0)$
		\item $t$ has coordinates $(1, 0)$
		\item all right endpoints of edges in $E_1$ (left endpoints of $E_2$ and $E_4$) have $x$-coordinate $\frac{1}{3}$
		\item all right endpoints of edges in $E_2$ (left endpoints of $E_3$) have $x$-coordinate $\frac{2}{3}$
		\item all edges in $E_1$, $E_2$, and $E_3$ are embedded as line segments
		\item all edges in $E_4$ are embedded so that they do not cross edges in $E_3$
	\end{itemize}
	\item Obtain $G_2$ by adding vertices at all edge crossings in the $x-y$ coordinate embedding of $G_1$. Note that all of the added crossings will have $x$-coordinate in the interval $(\frac{1}{3}, \frac{2}{3})$. Each edges in $G_2$ has a {\em parent edge} in $G_1$. For edges in $E_1$ and $E_3$, there are no crossings, and each edge in $E_1$ and $E_3$ has exactly one {\em child edge} in $G_2$.  The added crossings splits an edge in $E_2$ and $E_4$ into multiple child edges. Thus, an each in $E_2$ and $E_4$ can possibly have multiple child edges. \label{CrossingConstruction}
	
	\item Now, construct $w_2$ as follows:\label{W2Construction}
	\begin{enumerate}
		\item Arbitrarily label the vertices of $G_1$ that correspond to edges of $G$ with the integers 1 through $|E(G)|$ inclusive.
		\item For any edge $e \in E(G_2)$ with its parent (a copy of itself) in $E_1$ (within $G_1$), let $w_2(e) = |E(G)|^2$.
		\item For any edge $e \in E(G_2)$ with its parent in $E_2$, let $w_2(e)$ be the label (assigned above in (a)) of the right endpoint of the parent of $e$ in $G_1$.
		\item For any edge $e \in E(G_2)$ with its parent (a copy of itself) in $E_3$, let $w_2(e)$ be the label (assigned above in (a)) of the left endpoint of the parent of $e$ in $G_1$.
		\item For any edge $e \in E(G_2)$ with its parent in $\{u, t\} \in E_4$, let $$w_2(e) = |E(G)|^2 - \sum_{\substack{\hat{e}\in E(G_2): \\ \hat{e} \in \delta^+(u), \kappa(\hat{e})\in E_2}} w_2(\hat{e})$$ where $\kappa: 2^{E(G_2)}\rightarrow 2^{E(G_1)}$ returns the parent edge of a child edge $G_2$.  See Fig.~\ref{MaxFlowNPCFigure}.
	\end{enumerate}
	\item Finally, construct $c_2$ as follows:\label{C2Construction}
	\begin{enumerate}
		\item Sort the crossing vertices added to $G_1$ in increasing order by $x$-coordinate (with ties broken arbitrarily) to obtain a list $\{v_i\}_{i = 1}^r$, where $r$ denotes the number of added points.
		\item For any edge $e\in E(G_2)$ with its parent in $E_1$, let $c_2(e) = w_2(e) = |E(G)|^2$.
		\item For any edge $e = \{v_i, v_j\}\in E(G_2)$ with its parent in $E_2$ or $E_4$, let $c_2(e) = (j - i)w_2(e)$.
		\item For any edge $e = \{u, v_i\}\in E(G_2)$ with its parent in $E_2$ or $E_4$ and $u$ having $x$-coordinate $\frac{1}{3}$, let $c_2(e) = iw_2(e)$.
		\item For any edge $e = \{v_i, w\}\in E(G_2)$ with its parent in $E_2$ and $w$ having $x$-coordinate $\frac{2}{3}$, let $c_2(e) = (r + 1 - i)w_2(e)$.
		\item For any edge $e = \{v_i, t\}\in E(G_2)$ with its parent in $E_4$, let $c_2(e) = (r + 2 - i)w_2(e)$.
		\item For any edge $e = \{u, w\}\in E(G_2)$ with its parent in $E_2$, $u$ with $x$-coordinate $\frac{1}{3}$, and $w$ with $x$-coordinate $\frac{2}{3}$, let $c_2(e) = (r + 1)w_2(e)$.
		\item For any edge $e = \{u, t\}\in E(G_2)$ with its parent in $E_4$ and $u$ with $x$-coordinate $\frac{1}{3}$, let $c_2(e) = (r + 2)w_2(e)$.
		\item For any edge $e = \{w, t\}\in E(G_2)$ with its parent in $E_3$ and $w$ with $x$-coordinate $\frac{2}{3}$, let $c_2(e) = w_2(e)$.
	\end{enumerate}
\end{enumerate}

\begin{proposition}\label{CostProp}
For any positive integer $k$, the total edge cost of a flow with value at least $k|E(G)|^2$ must be at least $(r + 3)k|E(G)|^2$.

\end{proposition}

\begin{proof}
In the embedding construction of $G_2$, there are with $r + 3$ cuts associated with the construction in order by $x$-coordinate. Let $\{z_i\}_{i = 1}^r$ be the $x$-coordinates of the crossings, $z_0 = \frac{1}{3}$, and $z_{r + 1} = \frac{2}{3}$. Let cut $C_i\subseteq E(G_2)$ for $i \in \{1, 2, \hdots, r + 1\}$ be the set of edges intersected by the line $x = \frac{z_{i - 1} + z_i}{2}$. Let $C_0 = E_1$ and $C_{r + 2} = E_3\cup E_4$.

The cost of an edge $e$ with capacity $w_2(e)$ in $G_2$ is assigned as $w_2(e)$ times the number of these $r + 3$ cuts that $e$ is in.  For any  $s-t$ path in $G_2$, the path intersects each of the $r+3$ cuts exactly once. Therefore, the cost of sending one unit of flow in $G_2$ is at least $(r+3)$, and sending  $k|E(G)|^2$ units of flow costs at least $(r+3) k|E(G)|^2$. This completes the proof.
\end{proof}

\begin{proposition}\label{PartialFlowProp}
For any positive integer $k$, the total edge cost of a flow with value at least $k|E(G)|^2$ that has partial flow along at least one edge must be at least $(r + 3)k|E(G)|^2 + 1$.
\end{proposition}

\begin{proof}
Let $e\in E(G_2)$ be an edge with partial flow under a flow $f: E(G_2)\rightarrow \mathbb{Z}_{\ge 0}$. Consider one of the $r + 3$ cuts discussed in the proof of Proposition \ref{CostProp} that contains $e$. It suffices to show that the cost of edges in this cut is at least $k|E(G)|^2 + 1$. Since the capacities and flows through edges are integers and flow on $e$ is less than its capacity, $c_2(e) - f(e)\ge 1$, which implies the result.
\end{proof}

\begin{proposition}\label{DistinctCapacitiesProp}
Let $G_2$ be the planar graph that results from edge crossing removal. Consider two edges $e, e'\in E(G_1)$ that cross in the $x-y$ embedding of $G_1$. Consider any edges $f, f'\in E(G_2)$ with parent edges $e$ and $e'$ respectively. Then, $w_2(f)\ne w_2(f')$.
\end{proposition}

\begin{proof}
Note that all children of $e$ have the same maximum capacity. The same holds for children of $e'$. Parts (a) and (c) of Part \ref{W2Construction} of the construction ensure that if $e, e'\in E_2$, then $w_2(f)\ne w_2(f')$. Note that no two edges in $E_4$ cross in the embedding of $G_1$. The only other possible scenario with a crossing occurs when $e\in E_2$ and $e'\in E_4$. It suffices to show that $w_2(f') \ge \frac{|E(G)|^2 - |E(G)|}{2} > |E(G)|$ if $|E(G)|\ge 4$ (if $|E(G)|\le 3$, we can solve the instance with brute force in constant time). This suffices because the weight of any edge with a parent in $E_2$ is at most $|E(G)|$.
\end{proof}

These three properties are important for proving the next lemma. This lemma implies Theorem \ref{MaxFlowNPC} since the reduction takes polynomial time, the edge weights are polynomially bounded in the size of the graph, and $G_2$ is planar.

\begin{lemma}
$G$ has an independent set of size at least $k$ if and only if there is a flow on $G_2$ with total edge cost at most $(r + 3)k|E(G)|^2$ with value $k|E(G)|^2$.
\end{lemma}

\begin{proof}
First, suppose that $G$ has an independent set $S\subseteq V(G)$ with $|S|\ge k$. Construct a flow $f: E(G_2)\rightarrow \mathbb{Z}_{\ge 0}$ as follows:

\begin{enumerate}
	\item Start by finding a flow $f': E(G_1)\rightarrow \mathbb{Z}_{\ge 0}$, where $G_1$ has a new capacity function $w_2': E(G_1)\rightarrow \mathbb{Z}_{\ge 0}$ defined by $w_2'(p) = w_2(e)$ for any child $e$ of $p$ (since all children have the same capacity). Let $S_1\subseteq V(G_1)$ be the vertices $G_1$ corresponding to vertices of $S$ in $G$ and let $T_1\subseteq V(G_1)$ be the vertices of $G_1$ corresponding to edges in $\delta(S)$ in $G$. Since $S$ is an independent set of $G$, the function

\begin{displaymath}
   f'(p) = \left\{
     \begin{array}{lr}
       w_2'(p) & : p \in \delta(S_1)\cup \delta^+(T_1)\\
       0 & : p \notin \delta(S_1)\cup \delta^+(T_1)
     \end{array}
   \right.
\end{displaymath}

is a flow on $G_1$ under the capacity function $w_2'$ with flow value $k|E(G)|^2$.
	\item Define $f$ by $f(e) = f'(p)$ for all $e\in E(G_2)$ with parent $p\in E(G_1)$.
\end{enumerate}

Note that $f$ also has flow value $k|E(G)|^2$. One can check that the flow $f$ has cost $(r + 3)k|E(G)|^2$ as well, completing this direction.

Now, suppose that $G_2$ has a flow $f: E(G_2)\rightarrow \mathbb{Z}_{\ge 0}$ with total edge cost $(r + 3)k|E(G)|^2$ and value $k|E(G)|^2$. By Proposition \ref{PartialFlowProp}, all edges either have no flow or full flow. Let $F\subseteq E(G_2)$ be the set of edges with full flow. We will start by showing that there is a flow $f'$ on $G_1$ with capacity function $w_2'$ such that for all edges $e\in E(G_2)$ with parent $p\in E(G_1)$, $f(e) = f'(p)$. Let the sequence of crossings as they were ordered for assigning edge costs in Part (\ref{C2Construction}a) be $\{x_i\}_{i = 1}^r$. For every $p\in E_2\cup E_4$, it suffices to show that its children have equal flow values under $f$. Consider a crossing $x_i$. Let $e_1, e_2$ be the incoming edges and let $e_1', e_2'$ be the outgoing edges, with $w_2(e_2') = w_2(e_2)$ and $w_2(e_1') = w_2(e_1)$. By Proposition \ref{DistinctCapacitiesProp}, $w_2(e_2) \ne w_2(e_1)$. There are four cases that must be dealt with, since all edges either have full flow or no flow with respect to $f$:

\begin{enumerate}
	\item $f(e_1) = f(e_2) = 0$. In this case, the only possible flow for $e_1'$ and $e_2'$ is $f(e_1') = f(e_2') = 0$.
	\item $f(e_1) = w_2(e_1)$ and $f(e_2) = w_2(e_2)$. This is full flow, so $f(e_1') = w_2(e_1)$ and $f(e_2') = w_2(e_2)$.
	\item $f(e_1) = w_2(e_1)$ and $f(e_2) = 0$. If any flow is sent along $e_2'$, then $e_1'$ has partial flow or $f(e_1') = 0$ and $e_2'$ has partial flow (since $w_2(e_1')\ne w_2(e_2')$). Therefore, no flow is sent on $e_2'$ and $f(e_2') = 0$, $f(e_1') = w_2(e_1)$.
	\item $f(e_1) = 0$ and $f(e_2) = w_2(e_2)$. It is not possible to send full flow along $e_1'$ in this case (similar to the previous case), so $f(e_1') = 0$ and $f(e_2') = w_2(e_2)$.
\end{enumerate}

In all cases, note that $f(e_1) = f(e_1')$ and $f(e_2) = f(e_2')$. Since this is true at all crossings, $f$ may be lifted to a flow $f': E(G_1)\rightarrow \mathbb{Z}_{\ge 0}$ with $f'(p) = f(e)$, where $p$ is the parent of $e\in E(G_2)$.\\

As in the proof of Lemma \ref{WeakMaxFlowNPC}, form an set $S\subseteq V(G)$ by taking the vertices of $G$ corresponding to left endpoints of edges with full flow under $f'$ in $G_1$. Note that $|S| = k$, since there must be exactly $k$ edges in $E_1$ with full flow and all other edges with no flow to give a flow value of $k|E(G)|^2$. It suffies to show that $S$ is an independent set of $G$. Suppose, instead, that there are two vertices $v_1, v_2\in S$ such that $\{v_1, v_2\}\in E(G)$. This implies that there is a vertex $u$ in $G_1$ corresponding to $\{v_1, v_2\}$ such that there are two incoming edges with full flow. However, this cannot occur, as this sum of the flows in these two edges will be double the capacity of the outgoing edge from $u$ in $E_3$. This is a contradiction, so $S$ must be an independent set.
\end{proof}

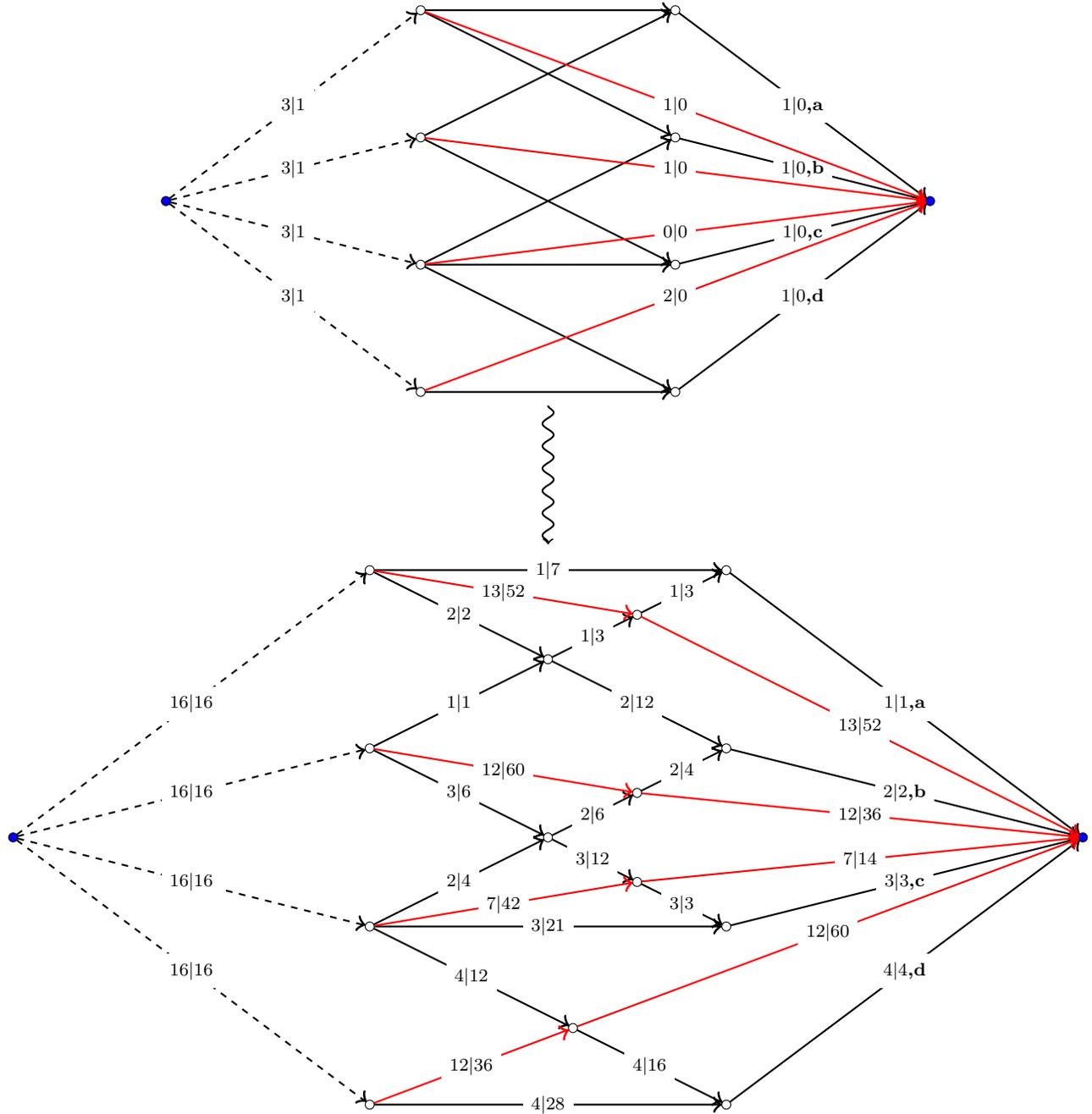
\begin{figure}[H]

\begin{tikzpicture}
\node(MaxFlowB){
	\begin{tikzpicture}
		\tikzset{VertexStyle/.append style = {fill = blue}}
		\Vertex[x = 0, y = 5]{S}
		\Vertex[x = 12, y = 5]{T}
		\tikzset{VertexStyle/.append style = {fill = none}}

		\Vertex[x = 4, y = 8]{A0}
		\Vertex[x = 4, y = 6]{A1}
		\Vertex[x = 4, y = 4]{A2}
		\Vertex[x = 4, y = 2]{A3}

		\Vertex[x = 8, y = 8]{a}
		\Vertex[x = 8, y = 6]{b}
		\Vertex[x = 8, y = 4]{c}
		\Vertex[x = 8, y = 2]{d}

		\tikzset{EdgeStyle/.style={ 
			postaction=decorate,
            		decoration={
                			markings,
                			mark=at position 1 with {\arrow[scale = 1.5]{>}}
            		}
		}}
		\Edge[style = dashed, label = $3|1$](S)(A0)
		\Edge[style = dashed, label = $3|1$](S)(A1)
		\Edge[style = dashed, label = $3|1$](S)(A2)
		\Edge[style = dashed, label = $3|1$](S)(A3)

		\Edge[label = {$1|0$,a}](a)(T)
		\Edge[label = {$1|0$,b}](b)(T)
		\Edge[label = {$1|0$,c}](c)(T)
		\Edge[label = {$1|0$,d}](d)(T)

		\Edge(A0)(a)
		\Edge(A1)(a)

		\Edge(A0)(b)
		\Edge(A2)(b)

		\Edge(A1)(c)
		\Edge(A2)(c)

		\Edge(A2)(d)
		\Edge(A3)(d)

		\Edge[color = red, label = $1|0$](A0)(T)
		\Edge[color = red, label = $1|0$](A1)(T)
		\Edge[color = red, label = $0|0$](A2)(T)
		\Edge[color = red, label = $2|0$](A3)(T)
	\end{tikzpicture}
};
\node[yshift = -10cm](MaxFlowC){
	\begin{tikzpicture}[scale = 1.4]
		\tikzset{VertexStyle/.append style = {fill = blue}}
		\Vertex[x = 0, y = 5]{S}
		\Vertex[x = 12, y = 5]{T}
		\tikzset{VertexStyle/.append style = {fill = none}}

		\Vertex[x = 4, y = 8]{A0}
		\Vertex[x = 4, y = 6]{A1}
		\Vertex[x = 4, y = 4]{A2}
		\Vertex[x = 4, y = 2]{A3}

		\Vertex[x = 8, y = 8]{a}
		\Vertex[x = 8, y = 6]{b}
		\Vertex[x = 8, y = 4]{c}
		\Vertex[x = 8, y = 2]{d}

		\Vertex[x = 6, y = 7]{X0}
		\Vertex[x = 6, y = 5]{X1}
		\Vertex[x = 6.28, y = 2.86]{X2}
		\Vertex[x = 7, y = 7.5]{X3}
		\Vertex[x = 7, y = 5.5]{X4}
		\Vertex[x = 7, y = 4.5]{X5}

		\tikzset{EdgeStyle/.style={ 
			postaction=decorate,
            		decoration={
                			markings,
                			mark=at position 1 with {\arrow[scale = 1.5]{>}}
            		}
		}}
		\Edge[style = dashed, label = $16|16$](S)(A0)
		\Edge[style = dashed, label = $16|16$](S)(A1)
		\Edge[style = dashed, label = $16|16$](S)(A2)
		\Edge[style = dashed, label = $16|16$](S)(A3)

		\Edge[label = {$1|1$,a}](a)(T)
		\Edge[label = {$2|2$,b}](b)(T)
		\Edge[label = {$3|3$,c}](c)(T)
		\Edge[label = {$4|4$,d}](d)(T)

		\Edge[label = $1|7$](A0)(a)
		\Edge[label = $1|1$](A1)(X0)
		\Edge[label = $1|3$](X0)(X3)
		\Edge[label = $1|3$](X3)(a)

		\Edge[label = $2|2$](A0)(X0)
		\Edge[label = $2|12$](X0)(b)
		\Edge[label = $2|4$](A2)(X1)
		\Edge[label = $2|6$](X1)(X4)
		\Edge[label = $2|4$](X4)(b)

		\Edge[label = $3|6$](A1)(X1)
		\Edge[label = $3|12$](X1)(X5)
		\Edge[label = $3|3$](X5)(c)
		\Edge[label = $3|21$](A2)(c)

		\Edge[label = $4|12$](A2)(X2)
		\Edge[label = $4|16$](X2)(d)
		\Edge[label = $4|28$](A3)(d)

		\Edge[color = red, label = $13|52$](A0)(X3)
		\Edge[color = red, label = $13|52$](X3)(T)
		\Edge[color = red, label = $12|60$](A1)(X4)
		\Edge[color = red, label = $12|36$](X4)(T)
		\Edge[color = red, label = $7|42$](A2)(X5)
		\Edge[color = red, label = $7|14$](X5)(T)
		\Edge[color = red, label = $12|36$](A3)(X2)
		\Edge[color = red, label = $12|60$](X2)(T)
	\end{tikzpicture}
};

\path[SquigglyArrow, ->](MaxFlowB)--(MaxFlowC);
\end{tikzpicture}

\caption{The result of crossing removal in the proof of Theorem \ref{MaxFlowNPC}. Labels are $w|c$. Note that six crossings were added, so $r = 6$ in this example. It costs at least 9 units of budget to send a unit of flow from $s$ to $t$. Furthermore, at any added crossing, the capacities of the incoming edges are distinct.}\label{MaxFlowNPCFigure}
\end{figure}

\section{Strong NP-hardness of shortest path interdiction on directed planar graphs}\label{ShortPathInt}

To prove Theorem \ref{ShortestPathNPC}, we reduce from BCFIP on planar graphs with source and sink adjacent to a face together, invoking Theorem \ref{MaxFlowNPC} to finish the proof. We use the max-flow min-cut theorem, the dual graph, and the fact that BCFIP is still strongly NP-complete when $s$ and $t$ are adjacent to one face to show that the max-flow problem can be transformed into a shortest path problem. The idea is similar to that of \cite{Phillips93} to show that the maximum flow interdiction problem is solvable in pseudo-polynomial time.

Consider a directed planar graph $G$, capacity function $w: E(G)\rightarrow \mathbb{Z}_{\ge 0}$, cost function $c: E(G)\rightarrow \mathbb{Z}_{\ge 0}$, budget $B > 0$, source vertex $s$, and sink vertex $t$ which collectively form an instance of BCFIP. Furthermore, assume that $s$ and $t$ are adjacent to a face together and $G\backslash s$ and $G\backslash t$ are connected (note that this is satisfied by the construction in Section \ref{MultiMaxFlow}). Without loss of generality, let this face be the infinite face of $G$. We will now construct a planar instance of DSPEIP as follows:

\begin{enumerate}
	\item Embed $G$ in the plane as $(\pi, E(G), E', R)$, where $E' = \{(u, v), (v, u): \forall (u, v)\in E(G)\}$ and $R((u, v)) = (v, u)$ for all $u, v\in V(G)$. Let $G_1 = G^*$ be the directed planar dual with respect to this embedding (see Figure \ref{DirectedPlanarDual}) and let $w\in V(G_1)$ be the vertex corresponding to the infinite face of $G$.

	\item Split $w$ into two vertices $u$ and $v$ to obtain the vertices of a graph $G_2$ and split the edges incident with $w$ as follows. Let $W = \{w_i\}_{i = 1}^k$ be a closed undirected walk in $G$ of vertices adjacent to the infinite face of $G$ that visits $s$ and $t$ exactly once with $w_1$ adjacent to $w_k$. Let $E(W)$ denote the set of edges on this walk. This walk exists because $G\backslash s$ and $G\backslash t$ are connected (see Proposition \ref{boundary-deletions}). Without loss of generality, suppose that $s$ appears before $t$ in $W$. Partition $E(W)$ into two sets $E_{st}$ and $E_{ts}$, where $E_{st}$ consists of edges that appear after $s$ and before $t$ in $W$ and $E_{ts} = E(W)\backslash E_{st}$. Have edges of $D_G(E_{st})$ be incident with $u$ instead of $w$ and have the edges of $D_G(E_{ts})$ be incident with $v$ instead of $w$. Let $g_2: E(G)\rightarrow E(G_2)$ be the edge set bijection described by this construction.

	\item Create a new graph $G_3$ from $G_2$ in which every edge is replaced with three copies, two in the same direction and one in the opposing direction. Let $g_{3a}: E(G)\rightarrow E(G_3)$ and $g_{3b}: E(G)\rightarrow E(G_3)$ injectively map edges of $G$ to same direction copies in $G_3$ and let $g_{3c}: E(G)\rightarrow E(G_3)$ injectively map edges of $G$ to reversed edges of $G_3$. Note that the sets $g_{3a}(E(G))$, $g_{3b}(E(G))$, and $g_{3c}(E(G))$ form a partition of $E(G_3)$.

	\item Let $w_3(g_{3a}(e)) = w(e)$ and $w_3(g_{3b}(e)) = w_3(g_{3c}(e)) = 0$ for each edge $e\in E(G)$.

	\item Let $c_3(g_{3a}(e)) = B + 1$, $c_3(g_{3b}(e)) = c(e)$, and $c_3(g_{3c}(e)) = B + 1$ for each edge $e\in E(G)$.
\end{enumerate}

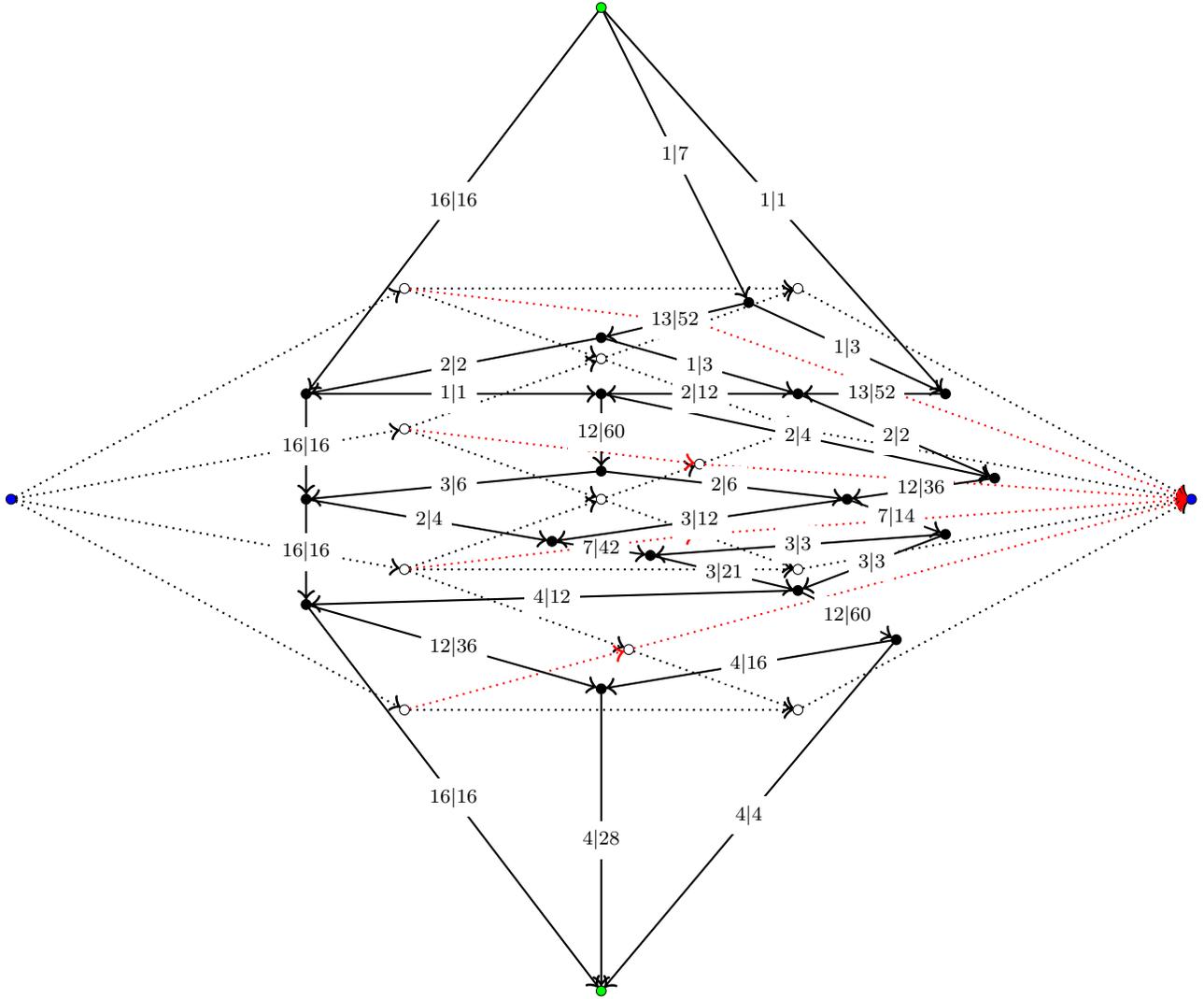
\begin{figure}
\begin{tikzpicture}[xscale = 1.4]

		\tikzset{VertexStyle/.append style = {fill = blue}}
		\Vertex[x = 0, y = 5]{S}
		\Vertex[x = 12, y = 5]{T}
		\tikzset{VertexStyle/.append style = {fill = none}}

		\Vertex[x = 4, y = 8]{A0}
		\Vertex[x = 4, y = 6]{A1}
		\Vertex[x = 4, y = 4]{A2}
		\Vertex[x = 4, y = 2]{A3}

		\Vertex[x = 8, y = 8]{a}
		\Vertex[x = 8, y = 6]{b}
		\Vertex[x = 8, y = 4]{c}
		\Vertex[x = 8, y = 2]{d}

		\Vertex[x = 6, y = 7]{X0}
		\Vertex[x = 6, y = 5]{X1}
		\Vertex[x = 6.28, y = 2.86]{X2}
		\Vertex[x = 7, y = 7.5]{X3}
		\Vertex[x = 7, y = 5.5]{X4}
		\Vertex[x = 7, y = 4.5]{X5}


		\tikzset{VertexStyle/.append style = {fill = green}}
		\Vertex[x = 6, y = 12]{U}
		\Vertex[x = 6, y = -2]{V}
		\tikzset{VertexStyle/.append style = {fill = black}}

		\Vertex[x = 3, y = 6.5]{B0}
		\Vertex[x = 3, y = 5]{B1}
		\Vertex[x = 3, y = 3.5]{B2}

		\Vertex[x = 7.5, y = 7.8]{B3}
		\Vertex[x = 6, y = 7.3]{B4}
		\Vertex[x = 6, y = 6.5]{B5}
		\Vertex[x = 6, y = 5.4]{B6}
		\Vertex[x = 5.5, y = 4.4]{B7}
		\Vertex[x = 6.5, y = 4.2]{B8}
		\Vertex[x = 6, y = 2.3]{B9}

		\Vertex[x = 9.5, y = 6.5]{B10}
		\Vertex[x = 8, y = 6.5]{B11}
		\Vertex[x = 10, y = 5.3]{B12}
		\Vertex[x = 8.5, y = 5]{B13}
		\Vertex[x = 9.5, y = 4.5]{B14}
		\Vertex[x = 8, y = 3.7]{B15}
		\Vertex[x = 9, y = 3]{B16}
		\tikzset{VertexStyle/.append style = {fill = none}}


		\tikzset{EdgeStyle/.style={ 
			postaction=decorate,
            		decoration={
                			markings,
                			mark=at position 1 with {\arrow[scale = 1.5]{>}}
            		}
		}}
		\Edge[style = dotted](S)(A0)
		\Edge[style = dotted](S)(A1)
		\Edge[style = dotted](S)(A2)
		\Edge[style = dotted](S)(A3)

		\Edge[style = dotted](a)(T)
		\Edge[style = dotted](b)(T)
		\Edge[style = dotted](c)(T)
		\Edge[style = dotted](d)(T)

		\Edge[style = dotted](A0)(a)
		\Edge[style = dotted](A1)(X0)
		\Edge[style = dotted](X0)(X3)
		\Edge[style = dotted](X3)(a)

		\Edge[style = dotted](A0)(X0)
		\Edge[style = dotted](X0)(b)
		\Edge[style = dotted](A2)(X1)
		\Edge[style = dotted](X1)(X4)
		\Edge[style = dotted](X4)(b)

		\Edge[style = dotted](A1)(X1)
		\Edge[style = dotted](X1)(X5)
		\Edge[style = dotted](X5)(c)
		\Edge[style = dotted](A2)(c)

		\Edge[style = dotted](A2)(X2)
		\Edge[style = dotted](X2)(d)
		\Edge[style = dotted](A3)(d)

		\Edge[style = dotted, color = red](A0)(X3)
		\Edge[style = dotted, color = red](X3)(T)
		\Edge[style = dotted, color = red](A1)(X4)
		\Edge[style = dotted, color = red](X4)(T)
		\Edge[style = dotted, color = red](A2)(X5)
		\Edge[style = dotted, color = red](X5)(T)
		\Edge[style = dotted, color = red](A3)(X2)
		\Edge[style = dotted, color = red](X2)(T)


		\Edge[label = $16|16$](U)(B0)
		\Edge[label = $1|7$](U)(B3)
		\Edge[label = $1|1$](U)(B10)
		\Edge[label = $16|16$](B0)(B1)
		\Edge[label = $1|1$](B0)(B5)
		\Edge[label = $16|16$](B1)(B2)
		\Edge[label = $2|4$](B1)(B7)
		\Edge[label = $12|36$](B2)(B9)
		\Edge[label = $16|16$](B2)(V)
		\Edge[label = $13|52$](B3)(B4)
		\Edge[label = $1|3$](B3)(B10)
		\Edge[label = $2|2$](B4)(B0)
		\Edge[label = $1|3$](B4)(B11)
		\Edge[label = $12|60$](B5)(B6)
		\Edge[label = $2|4$](B5)(B12)
		\Edge[label = $3|6$](B6)(B1)
		\Edge[label = $2|6$](B6)(B13)
		\Edge[label = $7|42$](B7)(B8)
		\Edge[label = $3|21$](B8)(B15)
		\Edge[label = $4|28$](B9)(V)
		\Edge[label = $13|52$](B10)(B11)
		\Edge[label = $2|12$](B11)(B5)
		\Edge[label = $2|2$](B11)(B12)
		\Edge[label = $12|36$](B12)(B13)
		\Edge[label = $3|12$](B13)(B7)
		\Edge[label = $7|14$](B13)(B14)
		\Edge[label = $3|3$](B14)(B8)
		\Edge[label = $3|3$](B14)(B15)
		\Edge[label = $4|12$](B15)(B2)
		\Edge[label = $12|60$](B15)(B16)
		\Edge[label = $4|16$](B16)(B9)
		\Edge[label = $4|4$](B16)(V)
\end{tikzpicture}
\caption{Reduction in the proof of Theorem \ref{ShortestPathNPC}. Dotted edges are edges from the old graph (from Figure \ref{MaxFlowNPCFigure}), while filled edges represent three different edges in the new graph. Every filled directed edge with weight $w|c$ in the new graph corresponds to one edge with weight $0|c$ in the same direction, one edge with weight $w|\infty$ in the same direction, and one edge with weight $0|\infty$ in the opposite direction.}\label{DirectedPlanarDual}
\end{figure}

Now, we prove Theorem \ref{ShortestPathNPC}:

\begin{proof}
Since the directed shortest path problem is in P, DSPEIP is in NP. Consider a graph $G$ with a capacity function $w: E(G)\rightarrow \mathbb{Z}_{\ge 0}$, cost function $c: E(G)\rightarrow \mathbb{Z}_{\ge 0}$, budget $B$, and source/sink nodes $s$ and $t$. By the max-flow min-cut theorem, solving BCFIP on $G$ is equivalent to solving the following problem on $G$:

\textbf{Given}: The directed graph $G$ with its capacity function $w$, cost function $c$, budget $B$, and source/sink nodes $s$ and $t$.

\textbf{Find}: A set of edges $I\subseteq E(G)$ with $c(I)\le B$ such that the total capacity of a minimum cut in $G[I]$ is maximized. More precisely, $I$ is a set that maximizes $$\min_{S\subseteq V(G): s\in S, t\notin S} w(\delta^+(S)\cap I) = \alpha(G[I], s, t)$$

Now, construct the directed graph $G_3$ by applying the previously described planar BCFIP to planar DSPEIP reduction with edge weight function $w_3: E(G_3)\rightarrow \mathbb{Z}_{\ge 0}$, interdiction cost function $c_3: E(G_3)\rightarrow \mathbb{Z}_{\ge 0}$, and budget $B$ (the same budget as for $G$).

\begin{center}
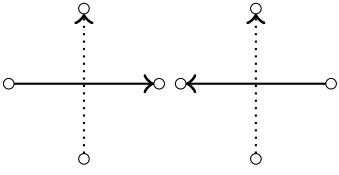
\begin{figure}
	\begin{tikzpicture}
		\Vertex[x = -1, y = 0]{W}
		\Vertex[x = 0, y = -1]{S}
		\Vertex[x = 1, y = 0]{E}
		\Vertex[x = 0, y = 1]{N}

		\tikzset{EdgeStyle/.style={ 
			postaction=decorate,
            		decoration={
                			markings,
                			mark=at position 1 with {\arrow[scale = 1.5]{>}}
            		}
		}}

		\Edge[style = dotted](S)(N)
		\Edge(W)(E)
	\end{tikzpicture}
	\begin{tikzpicture}
		\Vertex[x = -1, y = 0]{W}
		\Vertex[x = 0, y = -1]{S}
		\Vertex[x = 1, y = 0]{E}
		\Vertex[x = 0, y = 1]{N}

		\tikzset{EdgeStyle/.style={ 
			postaction=decorate,
            		decoration={
                			markings,
                			mark=at position 1 with {\arrow[scale = 1.5]{>}}
            		}
		}}

		\Edge[style = dotted](S)(N)
		\Edge(E)(W)
	\end{tikzpicture}
\caption{The first crossing diagram denotes the \emph{right orientation}, while the second crossing diagram denotes the \emph{left orientation}. Dotted edges are in $G$, while full edges are in $G_3$. Note that for every edge $e\in E(G)$, $g_{3a}(e)$ and $g_{3b}(e)$ cross it in the right orientation, while $g_{3c}(e)$ crosses $e$ in the left orientation.}\label{CrossingOrientations}
\end{figure}
\end{center}

First, we show that $\rho(G_3\backslash g_{3b}(I), u, v) \ge \alpha(G[I], s, t)$ for any $I\subseteq E(G)$. For this, it suffices to show that any $u-v$ shortest path in $G_3\backslash g_{3b}(I)$ from $u$ to $v$ can be made into an $s-t$ cut in $G[I]$ with the same weight. Consider a directed path $P\subset E(G_3)\backslash g_{3b}(I)$. Consider a simultaneous embedding of $G$ and $G_3$ in which no two edges of $G$ cross, no two edges of $G_3$ cross, and every edge $e\in E(G)$ crosses $g_{3a}(e), g_{3b}(e)$, and $g_{3c}(e)$. This also implies that every edge $g_{3a}(e)\in E(G_3)$ crosses $e\in E(G)$ (similarly for $g_{3b}$ and $g_{3c}$). In this embedding, edges of $G_3$ cross edges of $G$ in two different orientations (see Figure \ref{CrossingOrientations}). Partition $P$ into two sets $P_{right}$ and $P_{left}$ by crossing orientation. Note that $P_{left}\subseteq g_{3c}(E(G))$. Therefore, $w_3(P_{left}) = 0$. The edges in $P_{right}$ are either in $g_{3a}(I)$ or $g_{3b}(E(G))$. If there is an edge $e\in E(G)\backslash I$ such that $g_{3a}(e)\in P$, then $P$ can be made shorter or kept the same length in $G_3$ by replacing $g_{3a}(e)$ with $g_{3b}(e)$. Therefore, we may assume that $P_{right}$ only contains edges in $g_{3a}(I)$ or $g_{3b}(E(G)\backslash I)$.

$D_G(g_{3a}^{-1}(P)\cup g_{3b}^{-1}(P)\cup g_{3c}^{-1}(P))$ is a cycle (not necessarily directed) that passes through $w$ in $G_1$. By Proposition \ref{cut-cycle-duality}, $g_{3a}^{-1}(P)\cup g_{3b}^{-1}(P)\cup g_{3c}^{-1}(P)$ is an $s-t$ cut. Since $w_3(g_{3b}(E(G)\backslash I)) = 0$ and $G[I]$ only contains edges in $I$, $w_3(P_{right}) = w(g_{3a}^{-1}(P_{right}))$ and removing $g_{3a}^{-1}(P_{right})\cup g_{3c}^{-1}(P_{left})$ from $G[I]$ disconnects $s$ and $t$ in $G[I]$. Therefore, $\rho(G_3\backslash g_{3b}(I), u, v) \ge \alpha(G[I], s, t)$.

To complete the proof of Theorem \ref{ShortestPathNPC}, it suffices to show that $\rho(G_3\backslash g_{3b}(I), u, v) \le \alpha(G[I], s, t)$ for any $I\subseteq E(G)$ such that $c(I)\le B$. We just need to show that any minimum $s-t$ cut in $G[I]$ can be made into a path of $G_3\backslash g_{3b}(I)$ from $u$ to $v$ with equal weight. A minimum $s-t$ cut in $G[I]$, say given by the set $S$ with $s\in S$, $t\notin S$, corresponds to a cycle in the undirected dual graph that passes through the vertex corresponding to the infinite face of $G[I]$ by Proposition \ref{cut-cycle-duality}. Splitting this cycle at the vertex corresponding to the infinite face into $u$ and $v$ shows that there is a directed path in $G_3\backslash g_{3b}(I)$ from $u$ to $v$ that crosses the edges in $\delta(S)\cap I$ and no other edges in $G[I]$. If this path crosses any edges $e\in E(G)\backslash I$, it can pass through $g_{3b}(e)$ if it crosses with right orientation or $g_{3c}(e)$ if it crosses with left orientation. When the path crosses any edges $e\in I$, it can pass through $g_{3a}(e)$ if it crosses with right orientation or $g_{3c}(e)$ if it crosses with left orientation. This will cause the path to have the same weight as the cut, showing that $\rho(G_3\backslash g_{3b}(I), u, v) \le \alpha(G[I], s, t)$.
\end{proof}

\section{Strong NP-hardness of minimum perfect matching interdiction on planar bipartite graphs}\label{MinPerfectInt}

To prove Theorem \ref{MinPerfectMatchingNPC}, we introduce a new reduction from the directed shortest path problem to the minimum perfect matching problem. Unlike the currently used reduction in \cite{Hoffman63}, this reduction preserves planarity. Like the reduction in \cite{Hoffman63}, this reduction produces a bipartite graph.

Given a directed planar graph $G$ (instance of DSPEIP) with weight function $w: E(G)\rightarrow \mathbb{Z}_{\ge 0}$, interdiction cost function $c:E(G)\rightarrow \mathbb{Z}_{\ge 0}$, interdiction budget $B > 0$, path start node $u$, and path end node $v$, we construct three planar graphs before obtaining an instance of MPMEIP:

\begin{enumerate}
	\item\label{Deg3Step} Construct a directed planar graph $G_1$ on which solving DSPEIP is equivalent to solving DSPEIP on $G$. All vertices of $G_1$ except for $u$ and $v$ have degree at most 3.
	\begin{enumerate}
		\item Attach leaves $u_1$ to $u$ and $v_1$ to $v$, with a directed edge from $u_1$ to $u$ and from $v$ to $v_1$.
		\item Replace any vertex $w\in V(G)$ with $|\delta(w)| > 3$ (including the added leaves in the previous step) with a counterclockwise directed cycle of $|\delta(w)|$ edges. Number the new vertices $w_1, \hdots, w_{\delta(w)}$ and suppose that the counterclockwise order of incoming and outgoing edges for $w$ is $e_1, \hdots, e_{\delta(w)}$. For the edge $e_i$, reassign the endpoint that is $w$ to $w_i$. Let this graph be $G_1$. Let $f: E(G)\rightarrow E(G_1)$ map edges of $G$ to the corresponding edge of $G_1$.
		\item If $e\in E(G)$, let $w_1(f(e)) = w(e)$. If $e\in E(G_1)\backslash f(E(G))$, let $w_1(e) = 0$.
		\item If $e\in E(G)$, let $c_1(f(e)) = c(e)$. If $e\in E(G_1)\backslash f(E(G))$, let $c_1(e) = B + 1$.
	\end{enumerate}
	See an illustration of this construction in Figure~\ref{Fig:deg3}.
	
	\item\label{PseudoLineGraph} Construct the line graph $L(G_1)$ and add two vertices $u_2$ and $v_2$ to the graph. Construct $G_2$ from this graph by adding an edges $\{u_2, e\}$ for the edge $e\in \delta^+(u_1)$ in $G_1$ and $\{e, v_2\}$ for the edge $e\in \delta^-(v_1)$ in $G_1$. Let $f_1: E(G_1)\rightarrow V(G_2)$ be the injection induced by this construction.

Note that directed paths from $u_1$ to $v_1$ in $G_1$ are mapped to undirected paths in $G_2$. Since $G_1$ only has vertices (besides $u_1$ and $v_1$) with maximum degree three, $G_2$ is planar by 
Proposition \ref{planar-directed-line-graph}.

See an illustration of this step in Figure~\ref{Fig:linegraph}

	\item\label{MatchingSubdivision} Construct a weighted undirected planar graph $G_3$ from $G_2$. Minimum perfect matchings in subgraphs obtainable by interdiction of $G_3$ correspond to shortest paths from $u_1$ to $v_1$ in subgraphs of $G_1$.
	\begin{enumerate}
		\item For every edge $e\in E(G_1)$, split $f_1(e)\in V(G_2)$ into four vertices $f_{2a}(f_1(e)), f_{2b}(f_1(e))$, $f_{2c}(f_1(e)), f_{2d}(f_1(e))\in V(G_3)$ and three edges $\{f_{2a}(f_1(e)), f_{2b}(f_1(e))\}$, $\{f_{2b}(f_1(e)), f_{2c}(f_1(e))\}$, $\{f_{2c}(f_1(e)), f_{2d}(f_1(e))\}\in E(G_3)$. Define $g_{2a}: f_1(E(G_1))\rightarrow E(G_3)$ by $$g_{2a}(f_1(e)) = \{f_{2a}(f_1(e)), f_{2b}(f_1(e))\}$$ Define $g_{2b}$ and $g_{2c}$ similarly. Endpoints of edges of $G_2$ are reassigned to $f_{2b}$, $f_{2c}$, or $f_{2d}$ of the corresponding endpoints.
		\item Let $w_3(g_{2b}(f_1(e))) = w_1(e)$ for all $e\in E(G_1)$. For all other edges $e\in E(G_3)$, let $w_3(e) = 0$.
		\item Let $c_3(g_{2b}(f_1(e))) = c_1(e)$ for all $e\in E(G_1)$. For all other edges $e\in E(G_3)$, let $c_3(e) = B + 1$.
        \item Rename $u_2$ and $v_2$ with $u_3$ and $v_3$ respectively.
	\end{enumerate}
See an illustration of this construction in Figure~\ref{Fig:split}
\end{enumerate}

\begin{center}
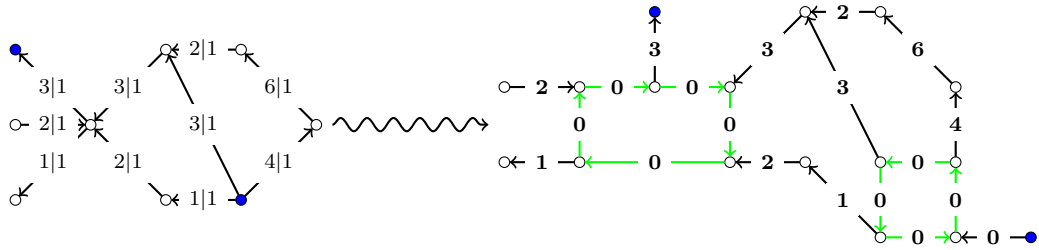
\begin{figure}
\begin{tikzpicture}
\node(MinPerfectMatchingFig1){
\begin{tikzpicture}
   \Vertex[x=0 ,y=0]{A0}
   \Vertex[x=0 ,y=1]{A1}
   \Vertex[x=1 ,y=1]{H0}
   \Vertex[x=2 ,y=2]{H1}
   \Vertex[x=3 ,y=2]{H2}
   \Vertex[x=4 ,y=1]{H3}
   \Vertex[x=2 ,y=0]{H5}
   \tikzset{VertexStyle/.append style = {shape = circle, fill = blue}}
   \Vertex[x=3 ,y=0]{H4}
   \Vertex[x=0 ,y=2]{A2}

   \tikzset{EdgeStyle/.style={->}}
   \Edge[label = $1|1$](H0)(A0)
   \Edge[label = $2|1$](A1)(H0)
   \Edge[label = $3|1$](H0)(A2)

   \Edge[label = $3|1$](H1)(H0)
   \Edge[label = $2|1$](H2)(H1)
   \Edge[label = $6|1$](H3)(H2)
   \Edge[label = $4|1$](H4)(H3)
   \Edge[label = $1|1$](H4)(H5)
   \Edge[label = $2|1$](H5)(H0)

   \Edge[label = $3|1$](H4)(H1)
\end{tikzpicture}
};
\node[xshift = 8cm](MinPerfectMatchingFig2){
\begin{tikzpicture}
   \Vertex[x=0 ,y=0]{A0}
   \Vertex[x=0 ,y=1]{A1}
   \Vertex[x=1 ,y=0]{H0a}
   \Vertex[x=1 ,y=1]{H0b}
   \Vertex[x=2 ,y=1]{H0c}
   \Vertex[x=3 ,y=1]{H0d}
   \Vertex[x=3 ,y=0]{H0e}
   \Vertex[x=4 ,y=2]{H1}
   \Vertex[x=5 ,y=2]{H2}
   \Vertex[x=6 ,y=1]{H3}
   \Vertex[x=4 ,y=0]{H5}
   \Vertex[x=5 ,y=0]{H4a}
   \Vertex[x=5 ,y=-1]{H4b}
   \Vertex[x=6 ,y=-1]{H4c}
   \Vertex[x=6 ,y=0]{H4d}
   \tikzset{VertexStyle/.append style = {shape = circle, fill = blue}}
   \Vertex[x=7 ,y=-1]{B0}
   \Vertex[x=2 ,y=2]{A2}

   \tikzset{EdgeStyle/.style={->}}
   \Edge[label = 1](H0a)(A0)
   \Edge[label = 2](A1)(H0b)
   \Edge[label = 3](H0c)(A2)

   \Edge[label = 3](H1)(H0d)
   \Edge[label = 2](H2)(H1)
   \Edge[label = 6](H3)(H2)
   \Edge[label = 4](H4d)(H3)
   \Edge[label = 1](H4b)(H5)
   \Edge[label = 2](H5)(H0e)

   \Edge[label = 3](H4a)(H1)

   \Edge[color = green, label = 0](H0a)(H0b)
   \Edge[color = green, label = 0](H0b)(H0c)
   \Edge[color = green, label = 0](H0c)(H0d)
   \Edge[color = green, label = 0](H0d)(H0e)
   \Edge[color = green, label = 0](H0e)(H0a)

   \Edge[color = green, label = 0](H4a)(H4b)
   \Edge[color = green, label = 0](H4b)(H4c)
   \Edge[color = green, label = 0](H4c)(H4d)
   \Edge[color = green, label = 0](H4d)(H4a)

   \Edge[label = 0](B0)(H4c)
\end{tikzpicture}
};
\path[SquigglyArrow, ->](MinPerfectMatchingFig1)--(MinPerfectMatchingFig2);
\end{tikzpicture}
\caption{Depiction of Step \ref{Deg3Step}, Section \ref{MinPerfectInt}. All added edges have edge weight $0|\infty$. Since $v$ is a leaf of $G$, no edges weight with weight 0 is added to make $v$ a leaf.}\label{Fig:deg3}
\end{figure}
\vspace{.125 in}

\begin{figure}
\begin{tikzpicture}
\node(MinPerfectMatchingFig1){
\begin{tikzpicture}[scale = .9]
   \Vertex[x=0 ,y=0]{A0}
   \Vertex[x=0 ,y=1]{A1}
   \Vertex[x=1 ,y=0]{H0a}
   \Vertex[x=1 ,y=1]{H0b}
   \Vertex[x=2 ,y=1]{H0c}
   \Vertex[x=3 ,y=1]{H0d}
   \Vertex[x=3 ,y=0]{H0e}
   \Vertex[x=4 ,y=2]{H1}
   \Vertex[x=5 ,y=2]{H2}
   \Vertex[x=6 ,y=1]{H3}
   \Vertex[x=4 ,y=0]{H5}
   \Vertex[x=5 ,y=0]{H4a}
   \Vertex[x=5 ,y=-1]{H4b}
   \Vertex[x=6 ,y=-1]{H4c}
   \Vertex[x=6 ,y=0]{H4d}
   \tikzset{VertexStyle/.append style = {shape = circle, fill = blue}}
   \Vertex[x=7 ,y=-1]{B0}
   \Vertex[x=2 ,y=2]{A2}

   \tikzset{EdgeStyle/.style={->}}
   \Edge[label = 1](H0a)(A0)
   \Edge[label = 2](A1)(H0b)
   \Edge[label = 3](H0c)(A2)

   \Edge[label = 3](H1)(H0d)
   \Edge[label = 2](H2)(H1)
   \Edge[label = 6](H3)(H2)
   \Edge[label = 4](H4d)(H3)
   \Edge[label = 1](H4b)(H5)
   \Edge[label = 2](H5)(H0e)

   \Edge[label = 3](H4a)(H1)

   \Edge[color = green, label = 0](H0a)(H0b)
   \Edge[color = green, label = 0](H0b)(H0c)
   \Edge[color = green, label = 0](H0c)(H0d)
   \Edge[color = green, label = 0](H0d)(H0e)
   \Edge[color = green, label = 0](H0e)(H0a)

   \Edge[color = green, label = 0](H4a)(H4b)
   \Edge[color = green, label = 0](H4b)(H4c)
   \Edge[color = green, label = 0](H4c)(H4d)
   \Edge[color = green, label = 0](H4d)(H4a)

   \Edge[label = 0](B0)(H4c)
\end{tikzpicture}
};
\node[xshift = 8cm, yshift = 0cm](MinPerfectMatchingFig2){
\begin{tikzpicture}
   \Vertex[x=0 ,y=0]{A0}
   \Vertex[x=0 ,y=1]{A1}
   \Vertex[x=1 ,y=0]{H0a}
   \Vertex[x=1 ,y=1]{H0b}
   \Vertex[x=2 ,y=1]{H0c}
   \Vertex[x=3 ,y=1]{H0d}
   \Vertex[x=3 ,y=0]{H0e}
   \Vertex[x=4 ,y=2]{H1}
   \Vertex[x=5 ,y=2]{H2}
   \Vertex[x=6 ,y=1]{H3}
   \Vertex[x=4 ,y=0]{H5}
   \Vertex[x=5 ,y=0]{H4a}
   \Vertex[x=5 ,y=-1]{H4b}
   \Vertex[x=6 ,y=-1]{H4c}
   \Vertex[x=6 ,y=0]{H4d}

   \Vertex[x=.5, y=0]{C0}
   \Vertex[x=.5, y=1]{C1}
   \Vertex[x=1, y=.5]{C2}
   \Vertex[x=1.5, y=1]{C3}
   \Vertex[x=2, y=0]{C4}
   \Vertex[x=2, y=1.5]{C5}
   \Vertex[x=2.5, y=1]{C6}
   \Vertex[x=3, y=.5]{C7}
   \Vertex[x=3.5, y=0]{C8}
   \Vertex[x=3.5, y=1.5]{C9}
   \Vertex[x=4.5, y=-.5]{C10}
   \Vertex[x=4.5, y=1]{C11}
   \Vertex[x=4.5, y=2]{C12}
   \Vertex[x=5, y=-.5]{C13}
   \Vertex[x=5.5, y=-1]{C14}
   \Vertex[x=5.5, y=0]{C15}
   \Vertex[x=5.5, y=1.5]{C16}
   \Vertex[x=6, y=-.5]{C17}
   \Vertex[x=6, y=.5]{C18}
   \Vertex[x=6.5, y=-1]{C19}

   \tikzset{VertexStyle/.append style = {shape = circle, fill = blue}}
   \Vertex[x=7 ,y=-1]{B0}
   \Vertex[x=2 ,y=2]{A2}

   \Edge(C5)(A2)
   \Edge(B0)(C19)

   \Edge(C0)(C4)
   \Edge(C1)(C3)
   \Edge(C2)(C3)
   \Edge(C2)(C4)
   \Edge(C3)(C5)
   \Edge(C3)(C6)
   \Edge(C4)(C7)
   \Edge(C4)(C8)
   \Edge(C6)(C7)
   \Edge(C7)(C9)
   \Edge(C8)(C10)
   \Edge(C9)(C11)
   \Edge(C9)(C12)
   \Edge(C10)(C13)
   \Edge(C11)(C15)
   \Edge(C12)(C16)
   \Edge(C13)(C14)
   \Edge(C13)(C15)
   \Edge(C14)(C17)
   \Edge(C15)(C17)
   \Edge(C16)(C18)
   \Edge(C17)(C18)
   \Edge(C17)(C19)
\end{tikzpicture}
};
\path[SquigglyArrow, ->](MinPerfectMatchingFig1)--(MinPerfectMatchingFig2);
\end{tikzpicture}
\caption{Depiction of Step \ref{PseudoLineGraph}, Section \ref{MinPerfectInt}}\label{Fig:linegraph}
\end{figure}
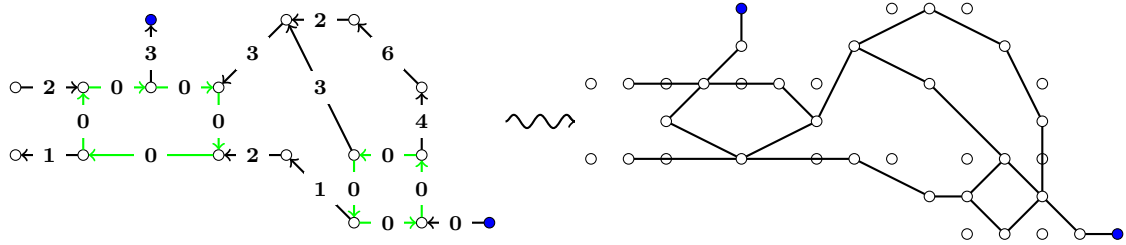
\vspace{.125 in}

\begin{figure}
\begin{tikzpicture}
\node(MinPerfectMatchingFig1){
\begin{tikzpicture}
   \Vertex[x=1, y=1.5]{C2}
   \Vertex[x=1, y=1]{C3}
   \Vertex[x=2.5, y=1]{C4}
   \Vertex[x=4, y=1.5]{I0}
   \Vertex[x=4, y=1]{I5}

   \Edge(C3)(C4)
   \Edge(C2)(C4)
   \Edge(C4)(I0)
   \Edge(C4)(I5)
\end{tikzpicture}
};
\node[xshift = 8cm, yshift = 0cm](MinPerfectMatchingFig2){
\begin{tikzpicture}
   \Vertex[x=1, y=1.5]{C2}
   \Vertex[x=1, y=1]{C3}
   \Vertex[x=2.5, y=1]{C4}
   \Vertex[x=3, y=1]{D0}
   \Vertex[x=4, y=1]{D1}
   \Vertex[x=4.5, y=1]{D2}
   \Vertex[x=6,y=1.5]{I0}
   \Vertex[x=6,y=1]{I5}

   \Edge(C3)(C4)
   \Edge(C2)(C4)
   \Edge(D2)(I0)
   \Edge(D2)(I5)
   \Edge(C4)(D0)
   \Edge[label = $0$](D0)(D1)
   \Edge(D1)(D2)
\end{tikzpicture}
};
\path[SquigglyArrow, ->](MinPerfectMatchingFig1)--(MinPerfectMatchingFig2);
\end{tikzpicture}
\caption{Depiction of one vertex subdivision in Step \ref{MatchingSubdivision}, Section \ref{MinPerfectInt}}\label{Fig:split}
\end{figure}
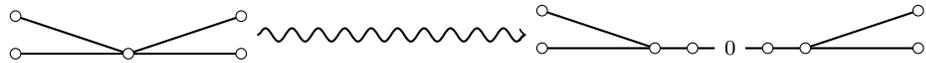
\end{center}

Since finding a minimum weight perfect matching in a graph is in P, MPMEIP is in NP. Note that the construction for Theorem \ref{ShortestPathNPC} has the property that for all interdiction sets $I\subseteq E(G)$ with $c(I)\le B$, there is a directed path from $u$ to $v$ in $G\backslash I$. We will now show that this construction suffices for proving strong NP-hardness of MPMEIP through three theorems.

\begin{theorem}
For any interdiction set $I\subseteq E(G)$ with $c(I)\le B$, there is a perfect matching in $G_3\backslash g_{2b}(f_1(f(I)))$ with the weight of a shortest path from $u$ to $v$ in $G\backslash I$. (i.e. $\mu(G_3\backslash g_{2b}(f_1(f(I))))\le \rho(G\backslash I, u, v)$).
\end{theorem}\label{LEMinPerfectMatching}

\begin{proof}
Let $P\subseteq E(G)$ be a shortest path from $u$ to $v$ in $G\backslash I$ (since a path from $u$ to $v$ exists in $G\backslash I$). Note that $f(P)$ is not a path from $u_1$ to $v_1$ in $G_1$ if $P$ passes through vertices with degree at least four. However, every vertex in $G$ with degree at least four is replaced with a directed cycle to obtain $G_1$. Therefore, for two consecutive edges $e, e'\in P$ with $e'$ following $e$ and shared endpoint $w$, there is a path with weight 0 from $t(f(e))$ to $s(f(e'))$. All of these added edges are in $G_1\backslash f(I)$, since $c_1(d) = B + 1$ for each added edge $d$. Since $P$ passes through every vertex of $G$ at most once, the result of adding segments from each directed cycle corresponding to a vertex of $G$ with degree at least four is a path. Call this path $Q\subseteq E(G_1)$. Since all added edges have weight zero, $w_1(Q) = w(P)$.

We now show that there is a perfect matching $N\subseteq E(G_3)\backslash g_{2b}(f_1(f(I)))$ with $w_3(N) = w_1(Q)$. For any $u-v$ simply path $L$ in $G_1$,  there exists the unique perfect matching in $G_3$ in which $g_{2b}(f_1(L)) \cup g_{2a}(f_1(E(G_1)\backslash L))\cup g_{2c}(f_1(E(G_1)\backslash L))$ is a part of the matching.
Note that $w_3(g_{2b}(f_1(Q))) = w_1(Q)$ and that $$M:= g_{2b}(f_1(Q))\cup g_{2a}(f_1(E(G_1)\backslash Q))\cup g_{2c}(f_1(E(G_1)\backslash Q))$$ is a matching in $G_3\backslash g_{2b}(f_1(f(I)))$. Furthermore, $M$ is a subset of a unique perfect matching $N\subseteq E(G_3)\backslash  g_{2b}(f_1(f(I)))$ with $w_3(N) = w_3(M) = w_1(Q) = w(P)$. Therefore, $N$ is the desired perfect matching.
\end{proof}

\begin{theorem}
For any interdiction set $I\subseteq E(G)$ with $c(I)\le B$, there is a path in $G\backslash I$ from $u$ to $v$ with the weight of a minimum weight perfect matching in $G_3\backslash g_{2b}(f_1(f(I)))$. (i.e. $\mu(G_3\backslash g_{2b}(f_1(f(I))))\ge \rho(G\backslash I, u, v)$).
\end{theorem}

\begin{proof}
It suffices to show that $\rho(G_1\backslash f(I), u_1, v_1) \le \mu(G_3\backslash g_{2b}(f_1(f(I))))$. Consider a perfect matching $M\subseteq E(G_3)\backslash g_{2b}(f_1(f(I)))$. $M$ exists by the proof of Theorem \ref{LEMinPerfectMatching}. Note that $w_1(f_1^{-1}(g_{2b}^{-1}(M))) = w_3(M)$ and that $f_1^{-1}(g_{2b}^{-1}(M))$ is a directed path from $u$ to $v$ in $G_1$ unioned with directed cycles. Let $C\subseteq E(G_1)$ be one of these directed cycles. Note that adding the edges $g_{2a}(f_1(C))\cup g_{2c}(f_1(C))$ to $M$ results in a cycle. Removing the edges of this cycle besides $g_{2a}(f_1(C))\cup g_{2c}(f_1(C))$ from $M$ (which includes $g_{2b}(f_1(C))$) and adding $g_{2a}(f_1(C))\cup g_{2c}(f_1(C))$ results in another perfect matching $N$ in $G_3\backslash g_{2b}(f_1(f(I)))$ with $w_3(N)\le w_3(M)$. This contradicts the fact that $M$ is a minimum perfect matching (unless $N$ has the same weight, in which we could have started out with it). Furthermore, $f_1^{-1}(g_{2b}^{-1}(N))$ has one fewer directed cycle. Therefore, there is a minimum perfect matching $K\subseteq E(G_3)\backslash g_{2b}(f_1(f(I)))$ with $f_1^{-1}(g_{2b}^{-1}(K))$ a directed path from $u_1$ to $v_1$. Therefore, we are done.
\end{proof}

\begin{theorem}
If $G$ is a directed planar graph, then $G_3$ is an undirected planar bipartite graph.
\end{theorem}

\begin{proof}
Replacing vertices with cycles preserves planarity. Since the vertices in $G_1$ have degree at most three, the construction to obtain $G_2$ preserves planarity. Finally, the construction to obtain $G_3$ preserves planarity, so all of this construction preserves planarity.

Now, it suffices to show that for any input directed graph $G$, $G_3$ is bipartite. For this, one can show that $G_3$ only has even cycles. Let $C\subseteq E(G_3)$ be a cycle in $G_3$. Note that $u_3$ and $v_3$ cannot be a part of $C$. Furthermore, $D := f_1^{-1}(g_{2a}^{-1}(C)\cup g_{2b}^{-1}(C)\cup g_{2c}^{-1}(C))$ is a cycle in the undirected version of $G_1$. Although $D$ may not be a directed cycle in $G_1$, the number of vertices on the cycle $D$ at which the direction of the edges reverses must be even.

Let $T\subseteq V(G_1)$ be the vertices of $D$ and let $S\subseteq T$ be the set of vertices at which $D$ switches direction. Every vertex $w$ in $G_1$ is associated with one or two edges in $G_3$ that connect $f_{2a}(e)$ and $f_{2d}(e')$ for edges $e, e'\in G_1$ with $t(e) = s(e') = w$. $D$ switches direction at $w$ only when $C$ passes through exactly two edges corresponding to $w$ in $G_3$. Also, $D$ maintains direction through $w$ only when $C$ passes through one edge corresponding to $w$ in $G_3$. Therefore, $|C| = 3|D| + |T| + |S| = 4|D| + |S|$. Since $|S|$ is even by the previous paragraph, $|C|$ is even, as desired.
\end{proof}

\section{Strong NP-hardness of maximum matching interdiction on planar bipartite graphs}\label{MaxInt}

We will prove Theorem \ref{MaxMatchingNPC} by edge weight manipulations. For an edge-weighted graph $G$ with a perfect matching, let $\gamma(G)$ be the maximum weight of any perfect matching in $G$. We will start by showing that the decision version of the following problem is strongly NP-complete on bipartite planar graphs:\\

\textbf{Input}: An edge-weighted graph $G = (V, E)$ with edge weight function $w: E\rightarrow \mathbb{Z}_{\ge 0}$, interdiction cost function $c: E\rightarrow \mathbb{Z}_{\ge 0}$, and interdiction budget $B > 0$.  It is assumed that, for every set $I\subseteq E$ with $c(I)\le B$, $G\backslash I$ has a perfect matching.\\

\textbf{Output}: The subset $I\subseteq E$ with $c(I)\le B$ that minimizes $\gamma(G\backslash I)$.\\

Now, we will reduce MPMEIP to this problem, which is NP because the maximum perfect matching problem is in P. Suppose we are given a graph $G$ with weight function $w: E(G)\rightarrow \mathbb{Z}_{\ge 0}$, cost function $c: E(G)\rightarrow \mathbb{Z}_{\ge 0}$, and interdiction budget $B > 0$ on which we want to solve MPMEIP. Let $W = \max_{e\in E(G)} w(e)$ and consider the weight function $w': E(G)\rightarrow \mathbb{Z}_{\ge 0}$ defined by $w'(e) = W - w(e)$ for all $e\in E(G)$. To prove strong NP-completeness, it suffices to show the following lemma (since edge weights are still bounded by a polynomial in the size of the graph):\\

\begin{lemma}
For any interdiction set $I\subseteq E(G)$ with $c(I)\le B$, $\gamma_{w'}(G\backslash I) = \frac{W|V(G)|}{2} - \mu_{w}(G\backslash I)$.
\end{lemma}

\begin{proof}
First, we will show that $\gamma_{w'}(G\backslash I)\le \frac{W|V(G)|}{2} - \mu_{w}(G\backslash I)$. By the definition of MPMEIP, $G\backslash I$ has a perfect matching. Let $M\subseteq E(G)\backslash I$ be one such perfect matching. Note that $|M| = \frac{|V(G)|}{2}$. Therefore, $w'(M) = \frac{W|V(G)|}{2} - w(M)$. Taking $M$ to be the maximum perfect matching with respect to $w'$ shows that $\gamma_{w'}(G\backslash I) = \frac{W|V(G)|}{2} - w(M) \le \frac{W|V(G)|}{2} - \mu_{w}(G\backslash I)$. Taking $M$ to be the minimum perfect matching with respect to $w$ shows that $\gamma_{w'}(G\backslash I)\ge w'(M) = \frac{W|V(G)|}{2} - \mu_{w}(G\backslash I)$, as desired.
\end{proof}

Finally, we will use this strong NP-completeness result to show that MMEIP is strongly NP-complete on planar graphs. Note that MMEIP is in NP because the maximum matching problem is in P. Let $w'': E(G)\rightarrow \mathbb{Z}_{\ge 0}$ be defined by $w''(e) = w'(e) + 2\nu_{w'}(G) + 2$. It suffices to show the following lemma, since the weights $w''$ are polynomially bounded in the size of the graph:\\

\begin{lemma}
For any interdiction set $I\subseteq E(G)$ with $c(I)\le B$, $\gamma_{w'}(G\backslash I) + |V(G)|\nu_{w'}(G) + |V(G)|= \nu_{w''}(G\backslash I)$.
\end{lemma}

\begin{proof}
Consider a perfect matching $M\subseteq E(G)\backslash I$. Note that $w'(M) + |V(G)|\nu_{w'}(G) + |V(G)| = w''(M)$ because $|M| = \frac{|V(G)|}{2}$. Therefore, if $M$ is the maximum weight perfect matching with respect to $w'$ in $G\backslash I$, then $\gamma_{w'}(G\backslash I) + |V(G)|\nu_{w'}(G) + |V(G)|= w''(M)\le \nu_{w''}(G\backslash I)$.\\

Now, suppose that $M$ is a maximum weight matching with respect to $w''$. Suppose, for the sake of contradiction, that $|M| \le \frac{|V(G)|}{2} - 1$. Then,

\begin{align*}
\nu_{w''}(G\backslash I) &= w''(M)\\
 &\le w'(M) +  (2\nu_{w'}(G) + 2)(\frac{|V(G)|}{2} - 1)\\
 &\le (|V(G)| - 1)\nu_{w'}(G) + |V(G)| - 2
\end{align*}

However, the weight of any perfect matching with respect to $w''$ is at least $(\nu_{w'}(G) + 1)|V(G)| > (|V(G)| - 1)\nu_{w'}(G) + |V(G)| - 2$, which contradicts the fact that $M$ is a maximum weight matching. Therefore, $|M| = \frac{|V(G)|}{2}$ and is perfect. Therefore,

\begin{align*}
\nu_{w''}(G\backslash I) &= w''(M)\\
&= w'(M) + (2\nu_{w'}(G) + 2)\frac{|V(G)|}{2}\\
&\le \gamma_{w'}(G\backslash I) + |V(G)|\nu_{w'}(G) + |V(G)|
\end{align*}

as desired.
\end{proof}

\section{Conclusion and open problems}\label{Conclusion}

In this paper, we described the complexity of edge interdiction problems when restricted to planar graphs. We presented a Pseudo-PTAS for the weighted maximum matching interdiction problem (MMEIP) on planar graphs. The algorithm extends Baker's Technique for local bilevel min-max optimization problems on planar graphs. Furthermore, we gave strong NP-hardness results for budget-constrained maximum flow (BCFIP), directed shortest path interdiction (DSPEIP), minimum perfect matching interdiction (MPMEIP), and maximum matching interdiction (MMEIP) on planar graphs. The latter three results followed from the strong NP-completeness of BCFIP on directed planar graphs with source and sink on a common face. This strong NP-completeness proof first reduced the maximum independent set problem to BCFIP on general directed graphs. We then noticed that optimal flows on these directed graphs had either full flow or no flow along all edges. To take advantage of this fact, we introduced a sweepline technique for assigning transportation costs and edge capacities in order to ensure that no two crossing edges had the same capacity. After introducing this embedding technique, edge crossings could be replaced with just one vertex.

There are many interesting open problems relating to interdiction. While hardness of approximation results are known for shortest path interdiction \cite{Khachiyan08}, it is only known that a $(2 - \epsilon)$-factor pseudo-polynomial time approximation would imply that $P = NP$. While heuristic solutions are known for the shortest path interdiction problem \cite{Israeli02}, no approximation algorithms with nontrivial approximation (or pseudoapproximation) guarantees are known. For the maximum matching interdiction problem, Zenklusen \cite{Zenklusen10} showed several hardness of approximation results. Nonetheless, no hardness of approximation results are known for MMEIP (computing $\nu_B^E(G)$ within a multiplicative factor). A pseudoapproximation is known for a continuous variant of the maximum flow interdiction problem \cite{Phillips03}. However, no algorithms and no hardness of approximation results are known for the (discrete) maximum flow interdiction problem on general directed graphs.

While we now know that BCFIP, directed shortest path interdiction, mimimum perfect matching interdiction, and MMEIP are strongly NP-complete on planar graphs, there are no known Pseudo-PTASes for BCFIP, shortest path interdiction, and minimum perfect matching interdiction when restricted to planar graphs.

\section*{Acknowledgements}
The authors are supported by Defense Threat Reduction Agency grant BRCALL08-A-2-0030, the Office of the Dean of the College at Princeton University, the Center for Nonlinear Studies at Los Alamos National Laboratory, and the DHS-STEM Research Fellowship Program. This research was performed under an appointment to the U.S. Department of Homeland Security (DHS) Science and Technology (S\&T) Directorate Office of University Programs HS-STEM Summer Internship Program, administered by the Oak Ridge Institute for Science and Education (ORISE) through an interagency agreement between the U.S. Department of Energy (DOE) and DHS. ORISE is managed by the Oak Ridge Associated Universities (ORAU) under DOE contract number DE-AC05-06OR23100. All opinions expressed in this paper are the author's and do not necessarily reflect the policies and views of DHS, DOE, and ORAU/ORISE.

\bibliographystyle{plain}
\bibliography{InterdictionRefs}

\end{document}